\documentclass[11pt,a4paper]{article}
\usepackage[margin=1in]{geometry}
\usepackage[T1]{fontenc} 
\usepackage{amsmath}
\usepackage{caption}

\usepackage[utf8]{inputenc}
\usepackage{bm}
\usepackage{verbatim}
\usepackage{float}
\usepackage{mwe}
\usepackage{color,soul}
\usepackage{threeparttable}
\usepackage{graphicx}
\usepackage{subcaption}
\usepackage{mwe}
\usepackage{mdframed}
\usepackage{xcolor}
\usepackage{authblk}
\usepackage{setspace}

\usepackage{natbib}
\bibpunct[, ]{(}{)}{,}{a}{}{,}%

\soulregister\cite7
\soulregister\ref7
\soulregister\eqref7
\soulregister\pageref7
\soulregister\equation7

\usepackage{multirow}
\usepackage{makecell}
\usepackage{booktabs}
\usepackage{array}
\usepackage{hhline}

\usepackage{mathtools}
\usepackage{bbm}
\usepackage{romanbar}
\usepackage{amsthm}
\usepackage{enumerate}
\usepackage{rotating} 
\usepackage{tikz}
\usetikzlibrary{matrix}
\usetikzlibrary{calc,intersections}
\usepackage{mathrsfs}

\usepackage{color}

\newtheorem{theorem}{Theorem}[section]

\newtheorem{lemma}[theorem]{Lemma}

\theoremstyle{definition}

\theoremstyle{remark}
\newtheorem{remark}[theorem]{Remark}
\numberwithin{equation}{section}

\newcommand{\bol}[1]{\mbox{\boldmath$#1$}}

\newcommand{\eqdist}{\stackrel{d}{=}}
\newcommand{\bSigma}{\bol{\Sigma}}

\newcommand{\bmu}{\bol{\mu}}

\newcommand{\btheta}{\bol{\theta}}

\newcommand{\bx}{\mathbf{x}}

\newcommand{\bM}{\mathbf{M}}

\newcommand{\bX}{\mathbf{X}}

\newcommand{\bw}{\mathbf{w}}

\newcommand{\bOne}{\mathbf{1}}

\newcommand{\E}{\mbox{E}}
\newcommand{\Var}{\mbox{Var}}

\newcommand{\bOmega}{\boldsymbol{\Omega}}

\newcommand{\bS}{\mathbf{S}}

\newcommand{\optn}[1]{\operatorname{#1}}
\newcommand{\VaR}{\operatorname{VaR}}
\newcommand{\CVaR}{\operatorname{CVaR}}

\newcommand{\argmin}{\mathop{\mathrm{argmin}}}

\providecommand{\keywords}[1]
{
\small	
\textbf{\textit{Keywords}:} #1
}

\begin{document}

\title{Bayesian Quantile-Based Portfolio Selection}

\author[1]{Taras Bodnar}
\author[1]{Mathias Lindholm}
\author[1]{Vilhelm Niklasson}
\author[1]{Erik Thorsén}
\affil[1]{Department of Mathematics, Stockholm University, SE-10691 Stockholm, Sweden}

\maketitle

\begin{abstract}
We study the optimal portfolio allocation problem from a Bayesian perspective using value at risk (VaR) and conditional value at risk (CVaR) as risk measures. By applying the posterior predictive distribution for the future portfolio return, we derive relevant quantiles needed in the computations of VaR and CVaR, and express the optimal portfolio weights in terms of observed data only. This is in contrast to the conventional method where the optimal solution is based on unobserved quantities which are estimated, leading to suboptimality. We also obtain the expressions for the weights of the global minimum VaR and CVaR portfolios, and specify conditions for their existence. It is shown that these portfolios may not exist if the confidence level used for the VaR or CVaR computation are too low. Moreover, analytical expressions for the mean-VaR and mean-CVaR efficient frontiers are presented and the extension of theoretical results to general coherent risk measures is provided. One of the main advantages of the suggested Bayesian approach is that the theoretical results are derived in the finite-sample case and thus they are exact and can be applied to large-dimensional portfolios. 

By using simulation and real market data, we compare the new Bayesian approach to the conventional method by studying the performance and existence of the global minimum VaR portfolio and by analysing the estimated efficient frontiers. It is concluded that the Bayesian approach outperforms the conventional one, in particular at predicting the out-of-sample VaR.
\end{abstract}


\section{Introduction}\label{sec:introduction}
The economic theory underlying optimal portfolio selection was pioneered by \cite{markowitz1952}. In his seminal paper, the optimal allocation of the available assets was determined by their expected returns and the covariance matrix of the asset returns. In practise, however, both the vector of expected returns and the covariance matrix are unknown and have to be estimated by using historical data. These estimates were traditionally treated as the true parameters of the data-generating process and plugged into the equations for the weights of optimal portfolios. Unfortunately, this causes a challenging problem since the optimal portfolio weights appear to be sensitive to misspecification of the input parameters, especially the expected returns, and thus estimation errors can lead to poorly performing portfolios \citep[see,][]{chopra1993effect, frankfurter1971portfolio, klein1976effect, merton1980estimating, simaan2014opportunity,  bodnar2018estimation}.

The problem concerning parameter uncertainty was one of the reasons why the Bayesian approach was introduced in portfolio theory during the 1970s \citep{winkler1975bayesian}. In the Bayesian setting, the parameters of the distribution of the asset returns are modelled as random variables with a prior distribution summarizing the information about the asset returns which is not present in the historical data. The posterior distribution, derived from the likelihood function and the used priors, provides updated knowledge on the model parameters conditionally on the observed data. In addition to reducing estimation risk, a Bayesian setting also makes it possible to employ useful prior information and implement fast numerical algorithms \citep[see,][]{avramov2010bayesian, rachev2008bayesian}. Several different ways of specifying a prior have been suggested in the literature \citep[see, e.g.,][]{avramov2010bayesian, BodnarMazurOkhrin2017, tu2010incorporating}. In general, a prior can be either informative or non-informative. Early applications of Bayesian statistics in finance were mainly based on uninformative or model-based priors \citep{bawa1979estimation}. However, during the 80s and 90s, some more sophisticated priors were developed. Two such highly influential priors are the hyperparameter prior by \cite{jorion1986bayes}, which was inspired by the Bayes-Stein shrinkage approach, and the informative prior by \cite{black1992global}, which relies on market equilibrium arguments. 

In addition to specifying a prior, an assumption must be made about the return distribution. Such a modelling choice is crucial both in the frequentist and in the Bayesian setup. The standard assumption of unconditional normality has been heavily criticized since actual market returns tend to have a higher density both close to the mean and far away from it \citep{cont2001empirical}. Because of this, more general assumptions allowing for heavier tails have been suggested \citep[see, e.g.,][]{adcock2014mean, bauder2018bayesian}. For instance, \cite{bauder2018bayesian} studied a Bayesian setting where they used the general assumption that the asset returns are exchangeable and multivariate centered spherically symmetric \citep[see also,][]{brugiere2020quantitative}. They derived a stochastic representation of the posterior predictive distribution, i.e., the distribution of the next return conditioned on the previous returns, under this assumption. By using this representation in the portfolio optimization problem, the common suboptimality issue in the standard plug-in approach was avoided since the estimates were no longer treated as the true parameters. Using a stochastic representation is a well-established tool in computational statistics \citep{givens2013computational} and it reduces the need for demanding Markov Chain Monte Carlo simulations. Moreover, a stochastic representation makes it easier to find quantiles which can be used to measure the risk. 

In the standard Markowitz setup, the risk is measured in terms of the variance of the portfolio return \citep{markowitz1952}. Recent regulations have forced banks and insurance companies to use other measures of risk, such as value at risk (VaR) and conditional value at risk (CVaR). The former is still in use in the Solvency II directive while the latter is enforced by the Basel III and IV standards. Both VaR and CVaR are quantile-based risk measures focusing on downside risk, meaning that the risk is determined from a quantile in the right tail of a loss distribution.  A lot of attention has been given to these two risk measures over the last decades concerning their applicability, computations and back-testing \citep[see, e.g.,][]{baumol1963expected, jorion1997value, pritsker1997evaluating, meng2020estimating, staino2020nested}. The two measures differ substantially in the way the quantiles are selected and it is usually argued that CVaR is a better measure of risk in comparison to VaR since it does not violate the desirable property of subadditivity \citep{artzner1999coherent}.

The usage of VaR and CVaR in portfolio optimization has also become increasingly popular over the last decades. For instance, \cite{rockafellar2000optimization} and \cite{babat2018computing} studied the minimization of CVaR in a portfolio using linear programming. \cite{ alexander2002economic, alexander2004comparison} derived the weights of a portfolio when using VaR or CVaR in the objective function in the portfolio optimization problem under the assumption of multivariate normally distributed asset returns. They noted that the optimal allocation strategy does not depend on whether the variance, VaR or CVaR is used as a risk measure. However, the portfolio that globally minimizes VaR or CVaR may not coincide with the one which globally minimizes the variance. This is illustrated in Figure \ref{fig:simulation_mean_variance_ef} where the mean-variance efficient frontier is plotted together with the locations of the global minimum variance (GMV), global minimum VaR (GMVaR) and global minimum CVaR (GMCVaR) portfolios based on weekly returns for 50 randomly selected stocks in the S\&P 500 index. In this figure, the efficient frontier and the optimal portfolios are all calculated using the conventional method, i.e., using sample estimates and treating them as the true parameters of the data-generating model. The fact that the GMVaR and GMCVaR portfolios are located above the GMV portfolio on the mean-variance efficient frontier motivates the study of these portfolios also for an investor who constructs the portfolio following the mean-variance analysis, since their expected returns determine the lower bounds for an investor who wants to be efficient also from a mean-VaR or mean-CVaR perspective. Given the new regulations, this should be highly desirable.

\begin{figure}[H]
\begin{mdframed}[backgroundcolor=black!10,rightline=false,leftline=false]
  \begin{subfigure}[b]{0.5\linewidth}
    \centering
    \includegraphics[width=0.95\linewidth]{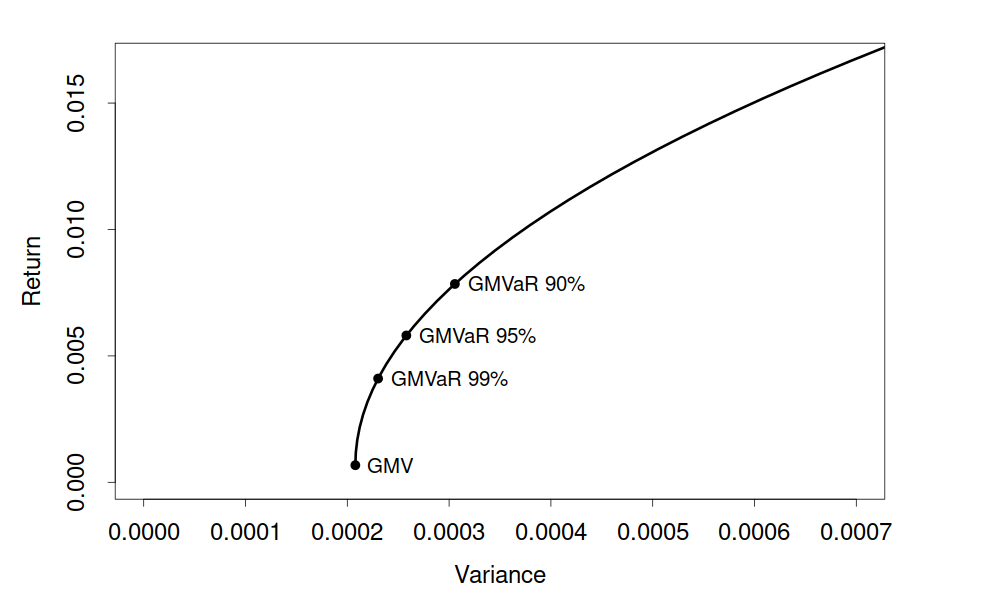}
    \caption{GMVaR}
  \end{subfigure}
  \begin{subfigure}[b]{0.5\linewidth}
    \centering
    \includegraphics[width=0.95\linewidth]{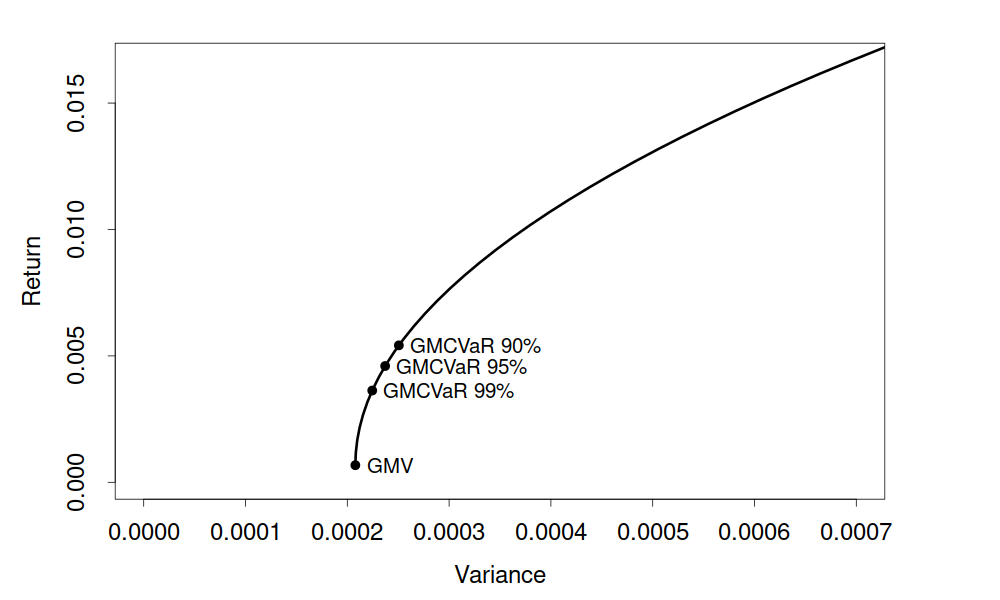}
    \caption{GMCVaR}
  \end{subfigure}
  \caption{Mean-variance efficient frontier based on empirical data from stocks in the S\&P 500 index together with the locations of the GMV, GMVaR and GMCVaR portfolios.}
  \label{fig:simulation_mean_variance_ef}
 \end{mdframed}
\end{figure}

This paper contributes to the current literature on portfolio theory by combining a Bayesian framework and a quantile-based asset allocation strategy. This gives new insights into how an optimal portfolio can be constructed in practice. We further compare the new Bayesian approach to its frequentist counterpart and demonstrate some of the advantages of the former through a simulation study and an empirical illustration. Throughout the report, we consider two different well-known priors (one informative and one non-informative) and we use general (non-normal) assumptions on the asset returns as  \cite{bauder2018bayesian} and \cite{brugiere2020quantitative}.

The rest of the paper is organized as follows. In Section 2, we present a Bayesian model of portfolio and asset returns and discuss its time series properties. Moreover, the posterior predictive distribution is derived together with stochastic representations related to it. In Section 3, we use the obtained stochastic representations to establish expressions for VaR and CVaR of a portfolio and we also extend the analysis to more general risk measures. Section 4 deals with portfolio optimization under parameter uncertainty. The Bayesian quantile-based portfolio optimization problem is solved and conditions for the existence of a global minimum VaR and CVaR portfolio are presented. The mean-VaR and mean-CVaR efficient frontiers are also derived. Section 5 contains a simulation study where our Bayesian approach is compared to the conventional method by analysing the performance and existence of the global minimum VaR portfolio as well as by comparing the estimated efficient frontiers. The comparison is continued in Section 6 with an empirical study based on market data from stocks in the S\&P 500 index. Section 7 contains a conclusion and discussion based on our findings. Finally, all of the proofs are moved to the appendix together with a description of the data.


\section{Bayesian model of portfolio and asset returns}\label{sec:bayesian_model_of_returns}
In this section, we introduce our Bayesian model and derive the posterior predictive distributions of the portfolio return using two different priors. We also make a time series interpretation of the model used to describe the stochastic properties of the asset returns.

\subsection{Posterior predictive distribution of portfolio return}\label{sec:posterior_predictive_dist}
Let $\bX_t$ be a $k$-dimensional vector of logarithmic asset returns at time $t$, i.e., each element $X_{t,i}$, $i=1,...,k$ of $\bX_t$ is defined as 
\begin{equation*}
    X_{t,i} \coloneqq \log\left(\frac{P_{t,i}}{P_{{t-1,i}}}\right),
\end{equation*}
where $P_{t, i}$ denotes the price at time $t$ of asset $i$. Throughout the paper we assume that the asset returns $\mathbf{X}_1,\mathbf{X}_2,...$ are infinitely exchangeable and multivariate centered spherically symmetric \citep[see,][for the definition and properties]{bernardo2009bayesian}. This assumption in particular implies that the asset returns are neither normally nor independently distributed. Moreover, let $\bx_{(t-1)}=(\bx_{t-n},...,\bx_{t-1})$ be the observation matrix of the asset returns $\bx_{t-n}$, ..., $\bx_{t-1}$ taken from time $t-n$ until $t-1$ whose distribution depends on the parameter vector $\btheta$. Bayes theorem provides the posterior distribution of $\btheta$ expressed as
\begin{equation*}
    \pi(\btheta|\bx_{(t-1)}) \propto f(\bx_{(t-1)}|\btheta)\pi(\btheta),
\end{equation*}
where $\pi(\cdot)$ is the prior distribution and $f(\cdot|\btheta)$ is the likelihood function. 

At time $t$ the return of the portfolio with weights $\mathbf{w}=(w_1,\dots,w_k)$, where it is assumed that $\mathbf{1}^\top \bw=1$, is given by
\begin{equation*}
    X_{P,t} \coloneqq \bw^\top\bX_t.
\end{equation*}
The posterior predictive distribution of $X_{P,t}$, i.e., the conditional distribution of $X_{P,t}$ given $\bx_{(t-1)}$, is computed by \citep[see,][]{bernardo2009bayesian}
\begin{equation}
    f(x_{P,t}|\bx_{(t-1)}) = \int_{\theta \in \Theta} f(x_{P,t}|\btheta)\pi(\btheta|\bx_{(t-1)})d\btheta,
    \label{eq:posterior}
\end{equation}
where $\Theta$ denotes the parameter space and $f(\cdot|\btheta)$ is the conditional density of $X_{P,t}$ given $\btheta$. The posterior predictive distribution \eqref{eq:posterior} determines the distribution of the future portfolio return at time $t$ given the information available in the historical data of asset returns up to $t-1$. It can be used to construct a point prediction of the future portfolio return, such as the posterior predictive mean or mode, together with the uncertainty, expressed as the posterior predictive variance. Furthermore, a prediction interval for the portfolio return can be obtained as a posterior predictive credible interval.

The integration in (\ref{eq:posterior}) makes it, usually, impossible to find an analytical expression for the posterior predictive distribution. In such cases, the point prediction of the future portfolio return together with its uncertainty can be obtained by using the Markov Chain Monte Carlo method \citep[see, e.g.,][]{kruger2020predictive}. Moreover, \cite{zellner2010direct} argued that a direct Monte Carlo approach provides an efficient computational method to compute a Bayesian estimation and to determine its uncertainty. Another way to deal with the integration issue is by directly simulating from the posterior predictive distribution by using a stochastic representation \citep[see, e.g., ][]{bauder2018bayesian}. 

Employing the non-informative Jeffreys prior and the informative conjugate prior,  \cite{bauder2018bayesian} derived a stochastic representation of a random variable following the posterior predictive distribution \eqref{eq:posterior} under the assumptions of infinitely exchangeability and multivariate centered spherical symmetry. In this case, the model parameters are the mean vector $\bm{\mu}$ and covariance matrix $\bm{\Sigma}$ of the asset returns (see, Section~\ref{sec:time_series_representation} for the detailed discussion). The Jeffreys prior and the conjugate prior are given by
\begin{equation}\label{eq:jeffreys_prior}
    \pi_J(\bmu,\bSigma) \propto |\bSigma|^{-(k+1)/2},
\end{equation}
and
\begin{equation}\label{eq:conjugate_prior}
    \bmu|\bm{\Sigma} \sim N_k\left(\mathbf{m}_0,\frac{1}{r_0}\bm{\Sigma}\right),
\quad \text{and} \quad
    \bm{\Sigma} \sim IW_k (d_0, \bS_0),
\end{equation}
respectively, where $|\cdot|$ stands for the determinant and $IW_k(d_0, \bS_0)$ denotes the inverse Wishart distribution with $d_0$ degrees of freedom and parameter matrix $\bS_0$ \citep[see,][for the definition and properties]{GuptaNagar2000}. The quantities $\mathbf{m}_0$, $r_0$, $d_0$, and $\bS_0$ are the hyperparameters of the conjugate prior where $\mathbf{m}_0$ and $\bS_0$ reflect the investor's prior belief about the mean vector and covariance matrix, whereas $r_0$ and $d_0$ represent the precision of these beliefs. The two priors \eqref{eq:jeffreys_prior} and \eqref{eq:conjugate_prior} have been widely used in financial literature where the Jeffreys prior is usually also referred to as the diffuse prior  \citep[see, e.g.,][]{Barry1974, BodnarMazurOkhrin2017, Brown1976, stambaugh1997analyzing}, and the conjugate prior is also commonly related to the Black-Litterman model \citep[see,][]{bauder2018bayesian, black1992global, frost1986empirical, kolm2017bayesian}.

The stochastic representations derived by \cite{bauder2018bayesian} fully determine the posterior predictive distribution. They are expressed in terms of two independent standard $t$-distributed random variables whose degrees of freedom depend on the assigned prior. Using these findings, the posterior predictive distribution is deduced and is presented in Theorem~\ref{theorem:predictive_posterior}. Let $t(q,a,b^2)$ denote the univariate $t$-distribution with $q$ degrees of freedom, location parameter $a$ and scale parameter $b$. Then, we obtain the following results.

\begin{theorem}\label{theorem:predictive_posterior}
Let asset returns $\mathbf{X}_1,\mathbf{X}_2,...$ be infinitely exchangeable and multivariate centered spherically symmetric. Then 
\begin{enumerate}
\item under the Jeffreys prior, it holds for $n>k$ that the posterior predictive distribution defined in \eqref{eq:posterior} is $t(d_{k,n,J},\mathbf{w}^T\bar{\bx}_{t-1,J}, r_{k,n,J}\mathbf{w}^T \mathbf{S}_{t-1,J} \mathbf{w})$ with $d_{k,n,J}=n-k$, 
\begin{equation}\label{eq:barx_J}
r_{k,n,J}=\frac{n+1}{n(n-k)},\quad \bar{\bx}_{t-1,J}=\frac{1}{n}\sum_{i=t-n}^{t-1}\bx_i,
\end{equation}
and
\begin{equation}\label{eq:S_j}
\mathbf{S}_{t-1,J} =\sum_{i=t-n}^{t-1}(\bx_i-\mathbf{\bar{x}}_{t-1,J})
(\bx_i-\mathbf{\bar{x}}_{t-1,J})^T;
\end{equation}
\item under the conjugate prior, it holds for $n+d_0-2k>0$ that the posterior predictive distribution defined in \eqref{eq:posterior} is $t(d_{k,n,C},\mathbf{w}^T\bar{\bx}_{t-1,C},r_{k,n,C}\mathbf{w}^T \mathbf{S}_{t-1,C} \mathbf{w})$ with $d_{k,n,C}=n+d_0-2k$,
\begin{equation}
r_{k,n,C}=\frac{n+r_0+1}{(n+r_0)(n+d_0-2k)}, \quad 
\bar{\bx}_{t-1, C} = \frac{n\bar{\bx}_{t-1,J}+r_0\mathbf{m_0}}{n+r_0},
\end{equation}
and
\begin{equation}
\bS_{t-1,C} = \bS_{t-1,J}+\bS_0+nr_0\frac{(\mathbf{m_0}-\bar{\bx}_{t-1, C})(\mathbf{m_0}-\bar{\bx}_{t-1, C})^\top}{n+r_0}.
\end{equation}
\end{enumerate}
\end{theorem}

The proof of Theorem~\ref{theorem:predictive_posterior} is given in the appendix. 

Since the structure of the posterior predictive distributions are similar, we introduce new notations which allow us to combine the two findings of Theorem~\ref{theorem:predictive_posterior} into a single result. Let
\begin{equation}\label{eq:notations_jeffreys_conjugate}
\left\{
\begin{matrix}
d_{k,n}=d_{k,n,J},\; r_{k,n}=r_{k,n,J},\;\bar{\bx}_{t-1}=\bar{\bx}_{t-1,J} , \;   \mathbf{S}_{t-1}=\mathbf{S}_{t-1,J}, & \text{under the Jeffreys prior,}\\
d_{k,n}=d_{k,n,C},\; r_{k,n}=r_{k,n,C},\; \bar{\bx}_{t-1}=\bar{\bx}_{t-1,C} , \;  \mathbf{S}_{t-1}=\mathbf{S}_{t-1,C}, & \text{under the conjugate prior.}
\end{matrix}
\right.
\end{equation}
From Theorem~\ref{theorem:predictive_posterior} we deduce that the posterior predictive distribution under both priors can be expressed as $t(d_{k,n},\mathbf{w}^T\bar{\bx}_{t-1},r_{k,n}\mathbf{w}^T \mathbf{S}_{t-1} \mathbf{w})$ with $d_{k,n}$, $r_{k,n}$, $\bar{\bx}_{t-1}$, and $\mathbf{S}_{t-1}$ as in \eqref{eq:notations_jeffreys_conjugate}. Let $\widehat{X}_{P,t}$ denote a random variable which follows the posterior predictive distribution, i.e., whose distribution coincides with the conditional distribution of portfolio return $X_{P,t}$ given the information available up to time $t-1$. Application of Theorem~\ref{theorem:predictive_posterior} together with \eqref{eq:notations_jeffreys_conjugate} then gives the following stochastic representation of $\widehat{X}_{P,t}$
\begin{equation}\label{eq:stochastic_rep_common}
\widehat{X}_{P,t} \eqdist \mathbf{w}^T\bar{\bx}_{t-1}+\tau\sqrt{r_{k,n}}\sqrt{\mathbf{w}^T\bS_{t-1}\mathbf{w}},
\end{equation}
where $\tau \sim t(d_{k,n})$ with $t(d_{k,n})$ denoting the standard $t$-distribution with $d_{k,n}$ degrees of freedom, i.e., the $t$-distribution with zero location parameter and scale parameter equal to one. Representation \eqref{eq:stochastic_rep_common} immediately implies that
\begin{equation} \label{eq:stochastic_rep_expectation_and_variance}
    \E(\widehat{X}_{P,t})=\E(X_{P,t}|\bx_{(t-1)})=\bw^\top\bar{\bx}_{t-1}, \; \Var(\widehat{X}_{P,t})=\Var(X_{P,t}|\bx_{(t-1)})=\frac{d_{k,n}r_{k,n}}{d_{k,n}-2}\bw^\top \bS_{t-1}\bw,
\end{equation}
for $d_{k,n}>1$ and $d_{k,n}>2$, respectively.

\subsection{Time series representation of asset returns}\label{sec:time_series_representation}

In the financial literature it is common to describe asset returns in terms of a time series model \citep[see, e.g.,][]{tsay2005analysis}. In order to make comparisons to such models easier, we now give a time series representation of the considered Bayesian model. 

The assumption of infinite exchangeability implies following De Finetti's theorem that there exists a probability measure conditionally on which the asset returns are independent and identically distributed \citep[see,][]{kingman1978uses}. Moreover, the assumption of multivariate centered spherically symmetry determines the conditional distribution of the asset returns. Namely, we get that the asset returns conditionally on the mean vector $\bol{\mu}$ and on the covariance matrix $\bSigma$ are independent and normally distributed \citep[see, Proposition 4.6 in][]{bernardo2009bayesian}. That is

\begin{equation}\label{eq:cond_model}
    \mathbf{X}_1,\mathbf{X}_2,...,\mathbf{X}_t|\bol{\mu},\bSigma \text{ are independent with}\quad \mathbf{X}_i|\bol{\mu},\bSigma \sim \mathcal{N}_k(\bol{\mu},\bSigma).
\end{equation}

Model \eqref{eq:cond_model} should not be confused with the assumption that the asset returns are independent and normally distributed as it is usually presented in financial and econometric literature. The two models coincide only when $\bmu$ and $\bSigma$ possess probability distributions concentrated in single points. This would imply that the investor knows  $\bmu$ and $\bSigma$ with certainty, which is not the case in practice.

In order to show the considerable difference between the model \eqref{eq:cond_model} and the model assuming independent and normally distributed asset returns, we next derive the marginal distribution of $\mathbf{X}_1$, $\mathbf{X}_2$, ..., $\mathbf{X}_t$ by using the findings of Theorem~\ref{theorem:predictive_posterior} and the stochastic representation \eqref{eq:stochastic_rep_common}. The results will be obtained by assigning the Jeffreys prior to the model parameters $\bmu$ and $\bSigma$ as well as the conjugate prior.

In Theorem~\ref{theorem:predictive_posterior} we showed that given a sample of size $n>k$, the posterior predictive distribution of $X_{P,t} = \bw^\top\bX_t$ given $\bX_{t-n},...,\bX_{t-1}$ has a univariate $t$-distribution under both priors (see, also the stochastic representation \eqref{eq:stochastic_rep_common}). Since $\bw$ is an arbitrary vector with the only restriction $\mathbf{1}^\top \bw=1$, choosing $n=t-1$ we get 
\begin{equation}\label{eq:dis_bX_t}
    \bX_t|\bX_{1},...,\bX_{t-1}\sim t_k(d_{k,t-1},\bar{\bx}_{1:t-1},r_{k,t-1}\bS_{1:t-1}),
\end{equation}
and, similarly for $j=k+2,....,t-1$ it holds that
\begin{equation}\label{eq:dis_bX_j}
    \bX_{j+1}|\bX_{1},...,\bX_{j}\sim t_k(d_{k,j},\bar{\bx}_{1:j},r_{k,j}\bS_{1:j}),
\end{equation}
where $d_{k,j}$ and $r_{k,j}$ for $j=k+2,...,t-1$ are defined as in \eqref{eq:notations_jeffreys_conjugate}; $\bar{\bx}_{1:j}$ and $\bS_{1:j}$ are defined similarly to $\bar{\bx}_{t-1}$ and $\bS_{t-1}$, where the sums in \eqref{eq:barx_J} and \eqref{eq:S_j} are from $1$ to $j$. The symbol $t_k(\cdot,\cdot,\cdot)$ stands for the $k$-dimensional multivariate $t$-distribution. 

The last result can also be written in the time series context in the following way
\begin{equation}\label{eq:ts_bX_j}
    \bX_{j+1}=\bar{\bx}_{1:j}+\sqrt{r_{k,j}}\bS_{1:j}^{1/2} \bol{\varepsilon}_j 
    \quad \text{for}\quad j=k+2,....,t-1,
\end{equation}
where $\bol{\varepsilon}_j \sim t_k(d_{k,j},\mathbf{0},\mathbf{I})$ is the error term and $\bS_{1:j}^{1/2}$ denotes a square root of $\bS_{1:j}$. Finally, we get that the joint density of $\bX_{1},...,\bX_{k+2}$ is given by \citep[see Appendix B in][]{bauder2020bayesian}
\begin{equation}
    f(\bx_1,...,\bx_{k+2}) \propto |\bS_{1:k+2}|^{-\frac{d_{k,k+2}+k-1}{2}}.
\end{equation}

The derived time series model for $\bX_{1},...,\bX_{t}$ indicates that a complicated nonlinear dependence structure is present in the conditional model \eqref{eq:cond_model}. Moreover, model \eqref{eq:cond_model} can be used to capture the time-dependent structure in the dynamics of both the conditional mean vector and the conditional covariance matrix of the asset returns which is usually observed in practice \citep[see, e.g.,][]{tsay2005analysis}. Finally, the error terms in the time series model for $\bX_{1},...,\bX_{t}$ in \eqref{eq:ts_bX_j} are multivariate $t$-distributed with degrees of freedom depending on the amount of available information. This gives rise to an additional source of non-stationarity and allows for heavier tails than in a model where asset returns are assumed to be independent and normally distributed corresponding to the conventional approach.


\section{Quantile-based risk measures}\label{sec:quantile_based_risk_measures}
The quantile based risk-measures value at risk (VaR) and conditional value at risk (CVaR) are introduced in the following section and we derive their analytical expressions in our Bayesian setting using the posterior predictive distribution. Moreover, we extend the analysis to more general risk measures.

\subsection{Posterior predictive VaR and posterior predictive CVaR}\label{sec:posterior_predictive_VaR_and_CVaR}

The results of Theorem~\ref{theorem:predictive_posterior} provides an easy way to sample from the posterior predictive distribution as well as to find its quantiles. The two most common quantile-based risk measures used in the literature are VaR and CVaR which we define in this section using the posterior predictive distribution of the portfolio return.

Denoting by $\widehat{X}_{P,t}$ a random variable whose distribution coincides with the posterior predictive distribution of the portfolio return $X_{P,t}$ given $\bx_{(t-1)}$ and using that the posterior predictive distribution is absolutely continuous, the two quantile-based risk measures at confidence level $\alpha \in (0.5,1)$ are defined by
$$\VaR_{\alpha,t-1}(\widehat{X}_{P,t}) := F_{Y,t-1}^{-1}(\alpha)$$ 
and 
$$\CVaR_{\alpha,t-1} (\widehat{X}_{P,t}) := \optn{E}[Y | Y \geq \VaR_{\alpha,t-1}(\widehat{X}_{P,t})],$$
respectively, where $Y := -\widehat{X}_{P,t}$ is the portfolio loss with cumulative distribution function $F_{Y,t-1}(\cdot)$.  

In the following, it is implicit that the probability and expectation in the definitions of the VaR and CVaR are formulated in terms of the posterior predictive distribution and they are conditioned on previous asset returns $\bx_{(t-1)}$. Moreover, the posterior predictive distribution is free of parameters and is fully determined by the observed data. This constitutes the main advantage of the suggested approach, namely it takes the parameter uncertainty into account before the optimal portfolio is constructed, i.e., the optimisation problem is formulated and solved. We discuss this point in detail in Section~\ref{sec:portfolio_optimization}.

By definition, $\VaR_{\alpha,t-1}(\widehat{X}_{P,t})$ satisfies
\begin{equation}\label{eq:VaR_defintion}
    P\left(\widehat{X}_{P,t} \le -\VaR_{\alpha,t-1}(\widehat{X}_{P,t})\right) = 1-\alpha.
\end{equation}
Rewriting the left hand side of \eqref{eq:VaR_defintion} yields
\begin{equation*}
\begin{split}
    P\left(\mathbf{w}^T\bar{\bx}_{t-1}+\tau\sqrt{r_{k,n}}\sqrt{\mathbf{w}^T\bS_{t-1}\mathbf{w}}\le -\VaR_{\alpha,t-1}(\widehat{X}_{P,t})\right) \\
    = P\left(\tau \le \frac{-\VaR_{\alpha}(\widehat{X}_{P,t}) - \mathbf{w}^T\bar{\bx}_{t-1}}
    {\sqrt{r_{k,n}}\sqrt{\mathbf{w}^T\bS_{t-1}\mathbf{w}}}\right),
 \end{split}
\end{equation*}
where $\tau$ is standard $t$-distributed with degrees of freedom $d_{k,n}$ defined in \eqref{eq:notations_jeffreys_conjugate} depending on the prior. Hence,
\begin{equation*}
    \frac{-\VaR_{\alpha,t-1}(\widehat{X}_{P,t}) - \mathbf{w}^T\bar{\bx}_{t-1}}
    {\sqrt{r_{k,n}}\sqrt{\mathbf{w}^T\bS_{t-1}\mathbf{w}}} = d_{1-\alpha},
\end{equation*}
where $d_{1-\alpha}$ is the $(1-\alpha)$ quantile of the $t$-distribution with $d_{k,n}$ degrees of freedom which satisfies $d_{\alpha}=-d_{1-\alpha}$ due to the symmetry of the $t$-distribution. Thus,
\begin{equation}\label{eq:VaR}
    \VaR_{\alpha,t-1}(\widehat{X}_{P,t})  = -\mathbf{w}^T\bar{\bx}_{t-1} +d_{\alpha}\sqrt{r_{k,n}}\sqrt{\mathbf{w}^T\bS_{t-1}\mathbf{w}}.
\end{equation}

Similarly, it follows by the definition of CVaR that
\begin{equation}\label{eq:CVaR}
    \CVaR_{\alpha,t-1}(\widehat{X}_{P,t}) = \E\left(-\widehat{X}_{P,t} |- \widehat{X}_{P,t}  \ge \VaR_{\alpha,t-1}(\widehat{X}_{P,t})\right)  
    =-\mathbf{w}^T\bar{\bx}_{t-1} +k_{\alpha}\sqrt{r_{k,n}}\sqrt{\mathbf{w}^T\bS_{t-1}\mathbf{w}}, 
\end{equation}
with
\begin{eqnarray*}
k_{\alpha}&=&\E\left(-\tau | -\tau \ge d_{\alpha}\right)
=\frac{1}{1-\alpha} \int_{d_{\alpha}}^{\infty} t f_{d_{k,n}}(t) \mbox{d}t
\\
&=&  \frac{1}{1-\alpha}\frac{\Gamma\left(\frac{d_{k,n}+1}{2}\right)}
    {\Gamma\left(\frac{d_{k,n}}{2}\right)
    \sqrt{\pi d_{k,n}}}
    \frac{d_{k,n}}{d_{k,n}-1}
    \left(1+\frac{d_{\alpha}^2}{d_{k,n}}\right)^{-\frac{d_{k,n}-1}{2}},
\end{eqnarray*}
where $f_{d_{k,n}}(t)$ denotes the density of the $t$-distribution with $d_{k,n}$ degrees of freedom and we use that the distribution of $-\tau$ coincides with $\tau$ due to the symmetry of the $t$-distribution. 

The expressions for VaR and CVaR given in \eqref{eq:VaR} and \eqref{eq:CVaR}, respectively, can be presented in the following form
\begin{equation}\label{eq:VaR_CVaR_common}
    Q_{t-1}(\bw) =-\mathbf{w}^T\bar{\bx}_{t-1} +q_{\alpha}\sqrt{r_{k,n}}\sqrt{\mathbf{w}^T\bS_{t-1}\mathbf{w}},
\end{equation}
where $q_{\alpha}=d_{\alpha}$ when considering VaR and $q_{\alpha}=k_{\alpha}$ when considering CVaR. This general formulation will be used extensively in the rest of the paper in order to handle VaR and CVaR cases simultaneously. 

Since $\alpha \in (0.5,1)$ and the $t$-distribution is symmetric, we get that $d_{\alpha}>0$. Moreover, we have that $k_{\alpha}>0$ by definition. As a result, $q_{\alpha}>0$ which together with the convexity of $\sqrt{\mathbf{w}^T\bS_{t-1}\mathbf{w}}$ implies that

\begin{theorem}\label{theorem:convexity_of_VaR_and_CVaR}
Under the conditions of Theorem~\ref{theorem:predictive_posterior}, $Q_{t-1}(\bw)$ is convex with respect to $\mathbf{w}$.
\end{theorem}

The proof of Theorem~\ref{theorem:convexity_of_VaR_and_CVaR} is given in the appendix.

\subsection{Extension to general risk measures}\label{sec:general_risk_measure}

The stochastic representation in equation \eqref{eq:stochastic_rep_common} shows that the general presentation of the posterior predictive VaR and of the posterior predictive CVaR as given in \eqref{eq:VaR_CVaR_common} can be extended to other risk measures used in portfolio theory and risk management. Such a result will provide a possibility to formulate and to solve the portfolio choice problem under more general setups considered in financial mathematics which is based on the context of coherent risk functionals \citep[see, e.g.,][]{artzner1999coherent}. Namely, instead of considering the VaR and the CVaR to specify the portfolio risk, one can choose any risk measure which is relevant, translation invariant and positive homogeneous. These properties are all present in the definition of a coherent risk measure. For a risk functional $\rho$, this can be formulated as
\begin{enumerate}
    \item [(a)] \textit{Relevance}: For all $X_{P,t} \leq 0$ we have $\rho(X_{P,t})\ge 0$. 
    \item [(b)] \textit{Translation invariance}: For any real scalar $a$, $\rho(X_{P,t} + a) = \rho(X_{P,t}) - a$.
    \item [(c)] \textit{Positive homogeneity}: For a positive scalar $\lambda$, $\rho(\lambda X_{P,t})=\lambda \rho(X_{P,t})$.
\end{enumerate}

The relevance property provides us with a foundation for the interpretation of risk measures. From an investor's point of view, a risk measure cannot assign a negative risk to a certain loss. That is often, if not always, how we think about risk. 
The property of translation invariance states that if one adds more cash to the portfolio or invest a larger proportion of the capital in something that is deemed as risk-free, then the risk-exposure should decrease by the corresponding amount. The property of positive homogeneity can be thought of as using leverage to acquire another (possibly more aggressive) position in the market. Such a position will increase the risk accordingly. 

When the asset returns $\mathbf{X}_1,\mathbf{X}_2,...$ are infinitely exchangeable and multivariate centered spherically symmetric, then the analytical expression of (posterior) general risk measures is deduced from \eqref{eq:stochastic_rep_common} and it is given by
\begin{align}
    \rho_{t-1}(\widehat{X}_{P,t}) 
        & = \rho_{t-1}(\mathbf{w}^T\bar{\bx}_{t-1}+\tau\sqrt{r_{k,n}}\sqrt{\mathbf{w}^T\bS_{t-1}\mathbf{w}}) \nonumber\\
        & = -\mathbf{w}^T\bar{\bx}_{t-1} +\rho_{t-1}(\tau)\sqrt{r_{k,n}}\sqrt{\mathbf{w}^T\bS_{t-1}\mathbf{w}},\label{gen_risk_measure}
\end{align}
which coincides with \eqref{eq:VaR_CVaR_common} when we set $\rho_{t-1}(\tau)=q_{\alpha}$. To this end, we note that $\rho_{t-1}(\tau)$ does not depend on the portfolio weights $\mathbf{w}$. As a result, finding the optimal portfolio by optimizing \eqref{gen_risk_measure} is the same problem as finding the optimal portfolio based on \eqref{eq:VaR_CVaR_common}. Since the conditions of relevance,  translation invariance and positive homogeneity are present in the definition of a coherent risk measure, any (posterior) coherent risk measure should satisfy \eqref{gen_risk_measure} and, consequently, the results of the next section can be applied.


\section{Portfolio optimization}\label{sec:portfolio_optimization}
We now turn to the theory related to portfolio optimization and present the solutions of the quantile-based optimal portfolio choice problems from both the Bayesian and the frequentist points of view. 

First, we describe existent results related to the conventional portfolio optimization problems. This is followed by theory related to the Bayesian quantile-based approach which relies on some previous findings. Special focus is put on the global minimum VaR and global minimum CVaR portfolios. The section ends with a derivation of the efficient frontier in the mean-quantile space.

\subsection{Portfolio optimization problems and existing solutions} \label{sec:conventional_portfolio_optimization}
The mean-variance optimization problem of \cite{markowitz1952} and its solution provide a fundamental concept to practical asset allocation. It was originally formulated in terms of the population parameters of the asset return distribution: mean vector $\bmu$ and covariance matrix $\bSigma$. As a result, the solution of Markowitz's optimization problem depends on these unknown quantities which have to be estimated before the implementation in practice. 

The population expected portfolio return of the portfolio with weights $\bw$ and its population variance are given by
$R_P(\bw)=\bw^\top \bol{\mu}$ and $V_P(\bw)=\bw^\top\bSigma\bw$, respectively. Then Markowitz's optimization problem is given by
\begin{equation}\label{eq:conventional_mv_problem}
\min_{\bw:\, R_P(\bw)=R_0,\,\bw^\top\mathbf{1}=1} V_P(\bw),
\end{equation}
where $\mathbf{1}$ denotes the $k$-dimensional vector of ones. The solution of \eqref{eq:conventional_mv_problem} is expressed as
\begin{equation}\label{eq:conventional_mv_solution}
\bw_{MV}=\bw_{GMV}+\frac{R_0-R_{GMV}}{s}\bM \bol{\mu}
\quad \text{with} \quad
\bM=\bSigma^{-1}-\frac{\bSigma^{-1}\mathbf{1}\mathbf{1}^\top\bSigma^{-1}}
{\mathbf{1}^\top\bSigma^{-1}\mathbf{1}},
\end{equation}
where
\begin{equation}\label{eq:conventional_weights_gmv}
\bw_{GMV}=\frac{\bSigma^{-1}\mathbf{1}}{\mathbf{1}^\top\bSigma^{-1}\mathbf{1}}
\end{equation}
are the weights of the global minimum variance (GMV) portfolio whose population expected return and population variance are given by
\begin{equation}\label{eq:conventional_return_and_variance_gmv}
R_{GMV}=\frac{\mathbf{1}^\top\bSigma^{-1}\bmu}{\mathbf{1}^\top\bSigma^{-1}\mathbf{1}}
\quad \text{and} \quad
V_{GMV}=\frac{1}{\mathbf{1}^\top\bSigma^{-1}\mathbf{1}}\,.
\end{equation}
The quantity $s$ is the slope parameter of the mean-variance efficient frontier, the set of all optimal portfolios in the mean-variance space, and it is given by
\begin{equation}
s=\bol{\mu}^\top \bM \bol{\mu}.
\end{equation}
The efficient frontier itself is a parabola given by \citep[see, e.g., ][]{Bodnar2009Econometricalanalysisofthesampleefficientfrontier,kan2008distribution, merton1972}
\begin{equation}\label{eq:conventional_mv_ef}
(R-R_{GMV})^2=s(V-V_{GMV}).    
\end{equation}
 
Recently, Makowitz's optimization problem has been reformulated from the Bayesian perspectives by \cite{bauder2018bayesian}. Using the portfolio expected return $R_{t-1}(\bw)=\E(X_{P,t}|\bx_{(t-1)})$ and the portfolio variance $V_{t-1}(\bw)=\Var(X_{P,t}|\bx_{(t-1)})$ computed from the posterior predictive distribution as in \eqref{eq:stochastic_rep_expectation_and_variance}, Markowitz's optimization problem from the Bayesian point of view is given by
\begin{equation}\label{eq:bayesian_mv_problem}
\min_{\bw:\, R_{t-1}(\bw)=R_0,\,\bw^\top\mathbf{1}=1} V_{t-1}(\bw),
\end{equation}
with the solution expressed as
\begin{equation}\label{eq:bayesianl_mv_solution}
\bw_{MV,t-1}=\bw_{GMV,t-1}+\frac{R_0-R_{GMV,t-1}}{s_{t-1}}\bM_{t-1} \bar{\bx}_{t-1}
\quad \text{with} \quad
\bM_{t-1}=\bS_{t-1}^{-1}-\frac{\bS_{t-1}^{-1}\mathbf{1}\mathbf{1}^\top\bS_{t-1}^{-1}}
{\mathbf{1}^\top\bS_{t-1}^{-1}\mathbf{1}},
\end{equation}
where
\begin{equation}\label{eq:bayesian_gmv_weights}
\bw_{GMV,t-1}=\frac{\bS_{t-1}^{-1}\mathbf{1}}{\mathbf{1}^\top\bS_{t-1}^{-1}\mathbf{1}}, \quad
R_{GMV,t-1}=\frac{\mathbf{1}^\top\bS_{t-1}^{-1}\bar{\bx}_{t-1}}
{\mathbf{1}^\top\bS_{t-1}^{-1}\mathbf{1}},
\quad \text{and} \quad
V_{GMV,t-1}=\frac{d_{k,n}r_{k,n}}{d_{k,n}-2}\frac{1}{\mathbf{1}^\top\bS_{t-1}^{-1}\mathbf{1}}\,.
\end{equation}
The quantity $s_{t-1}$ is one of the factors determining the slope of the Bayesian efficient frontier in the mean-variance space and it is given by
\begin{equation}\label{eq:bayesian_slope}
s_{t-1}=\bar{\bx}_{t-1}^\top\bM_{t-1}\bar{\bx}_{t-1}.
\end{equation}
Also from the Bayesian perspective, the efficient frontier is a parabola expressed as \citep[see,][]{bauder2018bayesian}
\begin{equation}\label{eq:bayesian_ef_mv}
(R-R_{GMV,t-1})^2=\frac{d_{k,n}-2}{d_{k,n}r_{k,n}}s_{t-1}(V-V_{GMV,t-1}).
\end{equation}

In contrast to the population optimal portfolios and the efficient frontier, the Bayesian optimal portfolio and the Bayesian efficient frontier are presented in terms of the historical data that are observable up to time $t-1$, when the optimal portfolio for the next period is constructed. Moreover, the Bayesian portfolio allocation is based on the predictive posterior distribution and incorporates the parameter uncertainty in the decision process before the weights of optimal portfolios are computed.

Assuming asset returns to be normally distributed, \cite{alexander2002economic, alexander2004comparison} extended Markowitz's optimization problem by replacing the population portfolio variance in \eqref{eq:conventional_mv_problem} with the population VaR and CVaR, respectively, given by
\begin{equation}\label{eq:VaR_CVaR_common_pop}
    Q_P(\bw) =-\mathbf{w}^T\bmu +q_{P;\alpha}\sqrt{\mathbf{w}^T\bSigma\mathbf{w}},
\end{equation}
where $q_{P;\alpha}=z_{\alpha}$ for the VaR and $q_{P;\alpha}=\dfrac{\exp\left(-z_{\alpha}^2 / 2\right)}{(1-\alpha)
\sqrt{2\pi}}$ for the CVaR where $z_{\alpha}$ denotes the $\alpha$-quantile of the standard normal distribution. 

The quantile-based optimization problems of \cite{alexander2002economic,alexander2004comparison} are given by
\begin{equation}\label{eq:conventional_min_Q}
\min_{\bw:\, R_{P}(\bw)=R_0,\,\bw^\top\mathbf{1}=1} Q_P(\bw).
\end{equation}
If the constraint on the expected return is omitted in \eqref{eq:conventional_min_Q}, then the solutions of \eqref{eq:conventional_min_Q} are the weights of the population optimal portfolios with the smallest values of VaR (or CVaR) at confidence level $\alpha$ given by \citep[see,][]{bodnar2012minimum}
\begin{equation}\label{eq:conventional_GMQ_weights}
\bw_{GMQ}=\bw_{GMV}+\frac{\sqrt{V_{GMV}}}{\sqrt{q_{P;\alpha}^2-s}}\bM \bol{\mu}.
\end{equation}
Similarly to the mean-variance portfolio, the weights \eqref{eq:conventional_GMQ_weights} of the population minimum VaR (or CVaR) portfolio cannot be computed. First, the unknown population parameters $\bmu$ and $\bSigma$ should be estimated by using historical data of asset returns and, then, the estimator of $\bw_{GMQ}$ is constructed as a proxy of the true portfolio weights. This two-step procedure of constructing an optimal portfolio usually leads to sub-optimal solutions since the parameter uncertainty is ignored in its construction. In the next subsection, we deal with the problem from the viewpoint of Bayesian statistics which allows to incorporate the parameter uncertainty directly into the decision process before the optimization problem is solved.

\subsection{Bayesian quantile-based optimal portfolios}\label{sec:bayesian_quantile_optimal_portfolios}

For a general quantile-based risk measure, the extension of the Alexander and Baptista optimization problem \eqref{eq:conventional_min_Q} from the Bayesian perspectives is given by
\begin{equation}\label{eq:bayesian_min_Q}
\min_{\bw:\, R_{t-1}(\bw)=R_0,\,\bw^\top\mathbf{1}=1} Q_{t-1}(\bw),
\end{equation}
where $R_{t-1}(\bw)$ and $Q_{t-1}(\bw)$ are computed by using the posterior predictive distribution as discussed in Section~\ref{sec:posterior_predictive_dist} and Section~\ref{sec:quantile_based_risk_measures}. For the special choices of the function $Q_{t-1}(\bw)$ as discussed in Section~\ref{sec:posterior_predictive_VaR_and_CVaR}, we get the optimization problems that minimize the predictive portfolio VaR and the predictive portfolio CVaR. 

The solution of the optimization problem \eqref{eq:bayesian_min_Q} can be presented in the following way
\begin{eqnarray*}
\argmin_{\bw:\, R_{t-1}(\bw)=R_0,\,\bw^\top\mathbf{1}=1} Q_{t-1}(\bw)
&=&
\argmin_{\bw:\, R_{t-1}(\bw)=R_0,\,\bw^\top\mathbf{1}=1} -R_{t-1}(\bw) + q_{\alpha}\sqrt{\frac{d_{k,n}-2}{d_{k,n}}}
\sqrt{V_{t-1}(\bw)}\\
&=&\argmin_{\bw:\, R_{t-1}(\bw)=R_0,\,\bw^\top\mathbf{1}=1} 
V_{t-1}(\bw),
\end{eqnarray*}
provided that $d_{k,n}>2$. Hence, on the one hand, all solutions of the quantile-based optimization problem \eqref{eq:bayesian_min_Q} are also the solutions of the mean-variance optimization problem \eqref{eq:bayesian_mv_problem} and belong to the efficient frontier \eqref{eq:bayesian_ef_mv}. On the other hand, all four optimization problems \eqref{eq:conventional_mv_problem}, \eqref{eq:bayesian_mv_problem}, \eqref{eq:conventional_min_Q}, and \eqref{eq:bayesian_min_Q} possess a solution only if $R_0$ is properly chosen. For example, \eqref{eq:conventional_mv_problem} and \eqref{eq:conventional_min_Q} have solutions if and only if $R_0>R_{GMV}$ and $R_0>R_{GMV,t-1}$, respectively, while for solving \eqref{eq:conventional_min_Q} one requires that
\begin{equation}\label{eq:conventional_GMQ_condition}
q_{P;\alpha}^2-s>0.    
\end{equation} Below in Theorem \ref{theorem:bayesian_GMQ}, we formulate the conditions needed for the existence of the Bayesian optimal portfolio in the sense of minimizing $Q_{t-1}(\bw)$. 

It has to be noted that the conditions of solution existence formulated in the case of the population optimization problems \eqref{eq:conventional_mv_problem} and \eqref{eq:conventional_min_Q} depend on the unknown population parameters of the data generating process and thus they cannot be validated in practice. In contrast, the Bayesian formulation of the optimization problems makes it possible to specify the existence conditions in terms of the previously observed data $\bx_{(t-1)}$ and, thus, to check them before the optimization problem is solved. Finally, the conditions on the existence of the solutions in the quantile-based optimization problems \eqref{eq:conventional_min_Q} and \eqref{eq:bayesian_min_Q} depend on the chosen confidence level $\alpha$, although the solutions themselves are independent of it.

Similarly to the mean-variance optimization problems, in order to determine under which conditions imposed on $R_0$ the solutions of the quantile-based optimization problem exist, one has to find the optimal portfolio with the smallest possible value of the objective function $Q_{t-1}(\bw)$, that is when the constraint $R_{t-1}(\bw)=R_0$ is dropped from the optimization problem \eqref{eq:bayesian_min_Q}. The expected return of this portfolio will provide the smallest possible value for which the optimization problem \eqref{eq:bayesian_min_Q} possesses a solution. To this end, we note that this is also the portfolio which a completely risk averse investor may be interested in. The following theorem expresses the variance and return of such a portfolio. 
\begin{theorem}\label{theorem:bayesian_GMQ}
Let $d_{k,n}>2$. Then, under the conditions of Theorem~\ref{theorem:predictive_posterior}, the global minimum quantile (GMQ)-based optimal portfolio exists if and only if
\begin{equation}\label{eq:bayesian_GMQ_condition}
q_{\alpha}^2>r_{k,n}^{-1}s_{t-1},
\end{equation}
where $s_{t-1}$ is defined in \eqref{eq:bayesian_slope}. Moreover, its posterior predictive expected return and variance are given by 
\begin{equation} \label{eq:return_GMVaR_GMCVaR}
R_{GMQ,t-1}=R_{GMV,t-1}+ \frac{r_{k,n}^{-1}s_{t-1}}{\sqrt{q_{\alpha}^2-r_{k,n}^{-1}s_{t-1}}}\sqrt{\frac{d_{k,n}-2}{d_{k,n}}}
\sqrt{V_{GMV,t-1}},
\end{equation}
and
\begin{equation} \label{eq:variance_GMVaR_GMCVaR}
    V_{GMQ,t-1} =  \frac{q_{\alpha}^2}{q_{\alpha}^2-r_{k,n}^{-1}s_{t-1}}V_{GMV, t-1},
\end{equation}
where $V_{GMV,t-1}$ and $R_{GMV,t-1}$ are given in \eqref{eq:bayesian_gmv_weights}.
\end{theorem}

The statement of Theorem~\ref{theorem:bayesian_GMQ} is proved in the appendix. Its results determine the lower bound for possible values of $R_0$ that can be used in the optimization problem \eqref{eq:bayesian_min_Q}. Since $R_{GMQ,t-1}>R_{GMV,t-1}$, we get that the set of optimal portfolios which solve \eqref{eq:bayesian_min_Q} does not coincide with the set of the Bayesian mean-variance optimal portfolios which lie on the upper part of the efficient frontier given by the parabola \eqref{eq:bayesian_ef_mv} in the mean-variance space. 

The findings of Theorem~\ref{theorem:bayesian_GMQ} lead to the expression of the smallest possible value of $Q_{t-1}(\bw)$ for the selected confidence level $\alpha$ expressed as
\begin{equation}\label{eq:bayesian_GMQ}
Q_{GMQ,t-1}=-R_{GMQ,t-1}+q_{\alpha}
\sqrt{\frac{d_{k,n}-2}{d_{k,n}}}\sqrt{V_{GMQ,t-1}}.
\end{equation}
Finally, the weights of the global minimum quantile-based portfolio are deduced from the findings of Theorem~\ref{theorem:bayesian_GMQ} and they are presented in Theorem~\ref{prop:weights_of_global_minimum_VaR_and_CVaR}. 
\begin{theorem}\label{prop:weights_of_global_minimum_VaR_and_CVaR}
Let $d_{k,n}>2$ and the inequality \eqref{eq:bayesian_GMQ_condition} holds. Then, under the conditions of Theorem~\ref{theorem:predictive_posterior}, the weights of the global minimum quantile-based optimal portfolio are given by
\begin{equation}\label{eq:bayesian_GMQ_weights}
    \bm{w}_{GMQ,t-1} = \bm{w}_{GMV,t-1}
    +\frac{r_{k,n}^{-1}\sqrt{V_{GMV,t-1}}}{\sqrt{q_{\alpha}^2-r_{k,n}^{-1}s_{t-1}}}
    \sqrt{\frac{d_{k,n}-2}{d_{k,n}}}
    \bM_{t-1}\bar{\bx}_{t-1}.
\end{equation}
\end{theorem} 

\subsection{Bayesian efficient frontier in mean-quantile space}\label{sec:bayesian_efficient_frontier}
Earlier in this section we proved that the solutions of the quantile-based portfolio optimization problem \eqref{eq:bayesian_min_Q} belong the Bayesian efficient frontier \eqref{eq:bayesian_ef_mv} in the mean-variance space. We now characterise the location of the Bayesian quantile-based optimal portfolio in the mean-quantile (mean-Q) space. It has to be noted that the population mean-VaR efficient frontier was carried out by \cite{alexander2002economic} under the assumption that the asset returns are multivariate normally distributed. We extend these findings in Theorem~\ref{theorem:bayesian_mean_Q_ef} whose proof is given in the appendix.

\begin{theorem}\label{theorem:bayesian_mean_Q_ef}
Let $d_{k,n}>2$ and $s_{t-1}>0$. Then, under the conditions of Theorem~\ref{theorem:predictive_posterior}, the Bayesian efficient frontier in the mean-Q space is a hyperbola given by
\begin{equation}\label{eq:bayesian_mean_Q_ef}
Q=q_{\alpha}\sqrt{\frac{(R-R_{GMV,t-1})^2}{r_{k,n}^{-1}s_{t-1}}+\frac{d_{k,n}-2}{d_{k,n}}V_{GMV,t-1}}-R.
\end{equation}
\end{theorem}

Expressions for the mean-variance efficient frontier using the Bayesian setup was derived in \cite{bauder2018bayesian}. It holds that the mean-variance efficient frontier \eqref{eq:bayesian_ef_mv} is a parabola in the mean-variance space and a hyperbola in the mean-standard deviation space for $s_{t-1}>0$. These findings are in line with the results in \cite{merton1972}, where the same conclusions were drawn for the population efficient frontier. In Theorem~\ref{theorem:bayesian_mean_Q_ef}, we prove that the efficient frontier in the mean-Q space is also a hyperbola under the same condition $s_{t-1}>0$. It is interesting to note that since $\bM_{t-1}$ is positive semi-definite with $\bM_{t-1}\bOne=\mathbf{0}$ by construction, it always holds that $s_{t-1}\ge0$ with $s_{t-1}=0$ only if the elements of the vector $\bar{\bx}_{t-1}$ are all equal. Another important observation is that both efficient frontiers \eqref{eq:bayesian_ef_mv} and \eqref{eq:bayesian_mean_Q_ef} are determined by the same set of quantities $R_{GMV,t-1}$, $V_{GMV,t-1}$, and $s_{t-1}$ which are computed from the historical data of asset returns.

\begin{remark}
Using the proof of Theorem~\ref{theorem:bayesian_mean_Q_ef}, we also obtain the analytical expression of the population efficient frontier in the mean-Q space, thus complementing the findings of \cite{alexander2002economic} who presents this frontier in the empirical study without deriving its closed-form expression. It holds that the population efficient frontier in the mean-Q space is a hyperbola expressed as
\begin{equation}\label{eq:conventional_mean_Q_ef}
Q=q_{P; \alpha}\sqrt{\frac{(R-R_{GMV})^2}{s}+V_{GMV}}-R.
\end{equation}
It is fully determined by the same set of constants $R_{GMV}$, $V_{GMV}$, and $s$ as the population efficient frontier \eqref{eq:conventional_mv_ef}, which is also a hyperbola in the mean standard-deviation space. 
\end{remark}

\section{Simulation study}\label{sec:simulation}
In the following section, we analyse how the Bayesian approaches compare to the conventional method via simulations. We will do so by studying VaR prediction using the global minimum VaR (GMVaR) portfolio and by looking at how frequently the conditions \eqref{eq:conventional_GMQ_condition} and \eqref{eq:bayesian_GMQ_condition} are satisfied. The comparison for other quantile-based risk measures can be done similarly by emphasising that the coherent risk measures have the same structure as given in \eqref{gen_risk_measure}. Utilizing \eqref{gen_risk_measure} it is interesting to note that any coherent risk measure can be rewritten as VaR at confidence level $\beta=F_{d_{k,n}}(\rho_{t-1}(\tau))$ for the Bayesian approaches, where  $F_{d_{k,n}}(\cdot)$ stands for the cumulative distribution function of the univariate $t$-distribution with $d_{k,n}$ degrees of freedom and $\tau$ is a $t$-distributed random variable with $d_{k,n}$ degrees of freedom. For the conventional method, one can use the same procedure, where the $t$-distribution is replaced by the standard normal distribution. Finally, we compare the different estimation methods by illustrating their corresponding efficient frontiers. When doing this comparison we also include the global minimum variance (GMV) portfolio in the analysis to see where it is located in the mean-VaR space.

\subsection{Setup of simulation study}
Throughout the simulation study, the asset returns are generated from a multivariate normal distribution. This distribution satisfies the assumptions of infinite exchangeability and multivariate centered spherical symmetry when conditioning on the parameters \citep[see, e.g., Proposition 4.6 in][]{bernardo2009bayesian}. In order to not restrict the analysis to certain parameters, the mean vector and covariance matrix are randomized in each new simulation iteration. We draw $\bm{\mu}$ from the uniform distribution on $[-0.003, 0.005]$, i.e., $\mu_i\sim U(-0.003, 0.005)$, and the covariance matrix is constructed by writing it as $\bm{\Sigma}=\mathbf{D}\mathbf{R}\mathbf{D}$ where $\mathbf{R}$ is a correlation matrix with $(\mathbf{R})_{ij}=0.3$ if $i \ne j$ and $\mathbf{D}$ is a diagonal matrix with entries given by $(\mathbf{D})_{ii} \sim U(0.03,0.04)$. We also consider different sample sizes and portfolio sizes by using $n\in\{100, 200\}$ and $k = cn$ for $c \in\{0.1, 0.3, 0.5, 0.7\}$. Moreover, we use $\alpha \in \{0.95, 0.99\}$ to study the impact of the confidence level in the GMVaR computations. For each parameter setup, we consider 10000 independent simulation runs when studying the performance and existence of the GMVaR portfolios. In each such simulation iteration, the out-of-sample performance of the portfolios is evaluated for one period ahead. We then aggregate the obtained results in all simulation runs. To this end, only those results where the conventional and Bayesian GMVaR conditions \eqref{eq:conventional_GMQ_condition} and \eqref{eq:bayesian_GMQ_condition} are satisfied simultaneously are considered.

Since the true parameters of the asset return distribution are known during simulation, it is possible to make comparisons with the population GMVaR portfolios as well as the population efficient frontier. The population GMVaR portfolios are constructed from the same mathematical formulas as when using the conventional method but they are based on the true parameters. Hence the population portfolios can be used as benchmarks for the corresponding Bayesian and conventional portfolios which are all based on parameter estimates. Similarly, the population efficient frontier can be used as a reference for the estimated efficient frontiers. Finally, the hyperparameters $\bm{m_0}$ and $\bm{S_0}$ of the conjugate prior are determined by using the empirical Bayesian approach \citep[see, e.g.,][]{bauder2020bayesian} where we set $d_0=r_0=n$.

\subsection{Existence and performance}\label{sec:sim_GMVaRP}
The GMVaR existence conditions \eqref{eq:conventional_GMQ_condition} and \eqref{eq:bayesian_GMQ_condition} are not satisfied in several simulation runs when the portfolio dimension becomes large in comparison to the sample size. More precisely, in the following cases $\{n = 100, k=50, \alpha=0.95 \}$, $\{n = 100, k=70, \alpha=0.95 \}$, $\{n = 200, k=100, \alpha=0.95 \}$, $\{n = 200, k=140, \alpha=0.95\}$, $\{n = 100, k=70, \alpha=0.99\}$ and $\{n = 200, k=140, \alpha=0.99\}$, the conventional condition \eqref{eq:conventional_GMQ_condition} is not met for 247, 8432, 1486, 9995, 1027 and 4155 out of the 10000 simulation runs, respectively. For the Jeffreys prior, the corresponding numbers are 0, 16, 0, 51, 0 and 0, and they are 4, 581, 4, 2694, 1 and 0 for the conjugate prior. For all other values of $\{n,k,\alpha\}$ the existence conditions are satisfied. They are always fulfilled for the population GMVaR portfolio. Based on these findings we conclude that both Bayesian GMVaR portfolios are more likely to exist than its conventional counterpart with the Bayeisan GMVaR portfolio under the Jeffreys prior demonstrating the lowest frequencies of non-existence in all of the considered cases. Our results also show that it is more likely that the GMVaR conditions are not satisfied when $\alpha$ is small or when $c$ is large. Finally, we point out that when $\{n = 200, k=140, \alpha=0.95\}$ the conventional GMVaR portfolio does not exist in 9995 out of 10000 simulation iterations. For this reason, the values of the performance measures for this configuration are not presented in Tables \ref{tab:simulation_GMVaRP_VaR_exceedance} and \ref{tab:simulation_GMVaRP_pop_VaR_deviation}.

We use two measures to analyze the performance of the GMVaR portfolios. The first measure is the relative frequency of times the estimated VaR is exceeded, i.e,
$$
    \frac{1}{N} \sum_{i=1}^N \mathbf{1}\{
        -X_{\text{GMVaRP}, i} \geq \widehat{\text{VaR}}_{\alpha}(X_{\text{GMVaRP}, i}) 
    \},
$$
where $N$ is the number of simulations, $\mathbf{1}$ is the indicator function, $X_{\text{GMVaRP}, i}$ is the actual return of the estimated GMVaR portfolio for simulation $i$ and $\widehat{\text{VaR}}_{\alpha}(X_{\text{GMVaRP}, i})$ is its predicted VaR. The latter two are calculated using equations \eqref{eq:bayesian_GMQ} and \eqref{eq:bayesian_GMQ_weights} in the Bayesian cases and \eqref{eq:VaR_CVaR_common_pop} and \eqref{eq:conventional_GMQ_weights} in the conventional and population cases. By the definition of VaR, an exceedance rate close to $1-\alpha$ means a good prediction of the VaR.  

In Table~\ref{tab:simulation_GMVaRP_VaR_exceedance} we observe that the relative VaR exceedance of the population GMVaR portfolios are always close to the target confidence level. The only source of noise is from the number of simulation runs. Such results do not hold for the three estimated GMVaR portfolios. The relative exceedance frequencies are close to the target confidence level when the portfolio dimension is small with respect to the sample size. For other values of $k$ and $n$, the predicted VaRs underestimate the true values. These results are in line with recent findings in portfolio theory, i.e., that the sample optimal portfolios are overoptimistic and tend to underestimate the risk. To this end we note that although the Bayesian approaches underestimate the risk, they still perform considerably better than the conventional approach. Especially, when $k/n \ge 0.5$ the relative exceedance rate for the conventional approach is almost twice as large as the ones obtained for the Bayesian methods when $\alpha=0.95$ and it is almost three times larger for $\alpha=0.99$. 

\begin{singlespace}
\begin{table}[H]
\small
\begin{mdframed}[backgroundcolor=black!10,rightline=false,leftline=false]
\centering
\caption{Relative VaR exceedance frequencies for the population GMVaR portfolio and its three estimates}
\label{tab:simulation_GMVaRP_VaR_exceedance}
{\begin{tabular}{| c | c | c | c | c |c | c |}
  \hline
  \multicolumn{3}{|c|}{Parameter setup}
  & \multicolumn{4}{c|}{GMVaR portfolio} \\
  \hline
  $\alpha$ & $n$ & $k$
   & Jeffreys & Conjugate & Conventional & Population\\
  \hhline{|=|=|=|=|=|=|=|}
    \multirow{8}{*}{$0.95$} & 
    \multirow{4}{*}{$100$} &
    $10$ & 0.0663 & 0.0741 & 0.0793 & 0.0510 \\ \cline{3-7}
    & & 
    $30$ & 0.1101 & 0.1375 & 0.1688 & 0.0510 \\ \cline{3-7}
    & & 
    $50$ & 0.1676 & 0.2237 & 0.3001 & 0.0490 \\ \cline{3-7}
    & &
    $70$ & 0.1996 & 0.2691 & 0.3967 & 0.0510 \\
  \cline{2-7}
    &\multirow{4}{*}{$200$} & 
    $20$ & 0.0640 & 0.0694 & 0.0760 & 0.0471 \\ \cline{3-7}
    & &
    $60$ & 0.1114 & 0.1389 & 0.1682 & 0.0490 \\ \cline{3-7}
    & &
    $100$ & 0.1649 & 0.2199 & 0.2984 & 0.0525 \\ \cline{3-7}
    & &
    $140$ & --- & --- & --- & --- \\
  \hline
   \multirow{8}{*}{$0.99$} & 
    \multirow{4}{*}{$100$} &
    $10$ & 0.0154 & 0.0190 & 0.0216 & 0.0114 \\ \cline{3-7}
    & & 
    $30$ & 0.0348 & 0.0510 & 0.0720 & 0.0095 \\ \cline{3-7}
    & & 
    $50$ & 0.0663 & 0.1049 & 0.1755 & 0.0098 \\ \cline{3-7}
    & &
    $70$ & 0.1160 & 0.1909 & 0.3300 & 0.0110\\
  \cline{2-7}
    &\multirow{4}{*}{$200$} & 
    $20$ & 0.0135 & 0.0158 & 0.0193 & 0.0091 \\ \cline{3-7}
    & &
    $60$ & 0.0330 & 0.0475 & 0.0667 & 0.0101 \\ \cline{3-7}
    & &
    $100$ & 0.0676 & 0.1074 & 0.1723 & 0.0115 \\ \cline{3-7}
    & &
    $140$ & 0.1158 & 0.1814 & 0.3179 & 0.0091 \\
    \hline
\end{tabular}}{} 
\end{mdframed}
\end{table}
\end{singlespace}

The relative VaR exceedance frequency measures how accurate the GMVaR estimates are by its definition. It does not measure how far away the estimated GMVaR is from the population GMVaR. To investigate this behaviour of our predictions we use a second measure which is the average absolute deviation of the estimated GMVaR to its population value, i.e.,
$$
    \frac{1}{N} \sum_{i=1}^N |
        \widehat{\text{VaR}}_{\alpha}(X_{\text{GMVaRP}, i}) - \text{VaR}_{\alpha}(X_{\text{GMVaRP}, i})
    |,
$$
where $\text{VaR}_{\alpha}(X_{\text{GMVaRP}, i})$ is the VaR of the population GMVaR portfolio. Since the population GMVaR portfolio is based on the true parameter values, the VaR of the population GMVaR portfolio coincides with true value of VaR. Hence, it can be used as a benchmark and the average absolute deviation should ideally be close to zero. 

Table~\ref{tab:simulation_GMVaRP_pop_VaR_deviation} shows the results of the GMVaR portfolio comparison using the average absolute deviation as a performance measure. Like in Table~\ref{tab:simulation_GMVaRP_VaR_exceedance}, we observe the same performance when the portfolio size is considerably smaller than the sample size. The Bayesian approach based on the Jeffreys prior shows the smallest deviations although the different portfolios are very close in their performance. For larger portfolio sizes, the Bayesian methods are significantly better than the conventional approach. They possess smaller values of the performance criteria as well as the computed standard deviations are smaller than those obtained for the conventional procedure. The differences become very large in the extreme case when $k/n=0.7$.  

\begin{singlespace} 
\begin{table}[H]
\small
\begin{mdframed}[backgroundcolor=black!10,rightline=false,leftline=false]
\centering
\caption{Average absolute deviation of the VaR of the estimated GMVaR portfolios to the VaR of the population GMVaR portfolio. Values inside the brackets represent the standard deviations.}
\label{tab:simulation_GMVaRP_pop_VaR_deviation}
{\begin{tabular}{| c | c | c | c | c |c |}
  \hline
  \multicolumn{3}{|c|}{Parameter setup}
  & \multicolumn{3}{c|}{GMVaR portfolio} \\
  \hline
  $\alpha$ & $n$ & $k$
   & Jeffreys & Conjugate & Conventional \\
  \hhline{|=|=|=|=|=|=|}
    \multirow{8}{*}{$0.95$} & 
    \multirow{4}{*}{$100$} &
    $10$ & 0.0027 (0.0021) & 0.0029 (0.0022) & 0.0032 (0.0024) \\ \cline{3-6}
    & & 
    $30$ & 0.0030 (0.0023) & 0.0048 (0.0030) & 0.0076 (0.0033)  \\ \cline{3-6}
    & & 
    $50$ & 0.0038 (0.0028) & 0.0079 (0.0038)  & 0.0148 (0.0042) \\ \cline{3-6}
    & &
    $70$ & 0.0038 (0.0029) &0.0084 (0.0039) & 0.0213 (0.0039)\\
  \cline{2-6}
    &\multirow{4}{*}{$200$} & 
    $20$ & 0.0019 (0.0014) & 0.0021 (0.0016) & 0.0026 (0.0018) \\ \cline{3-6}
    & &
    $60$ & 0.0023 (0.0017) & 0.0045 (0.0023)  & 0.0075 (0.0024) \\ \cline{3-6}
    & &
    $100$ & 0.0030 (0.0022) & 0.0078 (0.0028) & 0.0157 (0.0035) \\ \cline{3-6}
    & &
    $140$ & --- & --- & ---  \\
  \hline
   \multirow{8}{*}{$0.99$} & 
    \multirow{4}{*}{$100$} &
    $10$ & 0.0035 (0.0027) & 0.0035 (0.0027) & 0.0041 (0.0029) \\ \cline{3-6}
    & & 
    $30$ & 0.0035 (0.0027) & 0.0051 (0.0034) & 0.0088 (0.0038) \\ \cline{3-6}
    & & 
    $50$ & 0.0040 (0.0031) & 0.0080 (0.0043) & 0.0161 (0.0041) \\ \cline{3-6}
    & &
    $70$ & 0.0052 (0.0039) & 0.0110 (0.0053) & 0.0266 (0.0054) \\
  \cline{2-6}
    &\multirow{4}{*}{$200$} & 
    $20$ & 0.0023 (0.0017) & 0.0025 (0.0018) & 0.0032 (0.0022)  \\ \cline{3-6}
    & &
    $60$ & 0.0024 (0.0018) & 0.0047 (0.0026) & 0.0085 (0.0027)\\ \cline{3-6}
    & &
    $100$ & 0.0030 (0.0022) & 0.0079 (0.0032) & 0.0158 (0.0031) \\ \cline{3-6}
    & &
    $140$ & 0.0034 (0.0025) & 0.0103 (0.0036) & 0.0266 (0.0040) \\
    \hline
\end{tabular}}{}
\end{mdframed}
\end{table}
\end{singlespace}

To conclude, the suggested Bayesian methods perform better at estimating the GMVaR than the conventional approach. Using the Jeffreys prior seems to give the best results, although using the conjugate prior is also beneficial compared to the conventional method. However, all of the methods are underestimating the GMVaR, especially when $c$ is large. This is indicated by VaR exceedance frequencies much higher than $1-\alpha$ and large deviations to the population GMVaR. However, even if none of the methods perform very well for large-dimensional portfolios, this is the situation where we see the greatest benefit of using the Bayesian approaches. This point is further studied in the next section where we investigate the influence of parameter uncertainty on the estimation of the whole mean-VaR efficient frontier.

\subsection{Comparison of efficient frontiers}\label{sec:comparison_efficient_frontiers_sim}
In order to get a better understanding of the impact of parameter uncertainty on the quantile-based portfolio selection, we use the theoretical findings of Section~\ref{sec:bayesian_efficient_frontier} and plot the population mean-VaR efficient frontier together with its three estimates in Figures \ref{fig:simulation_mean_VaR_ef_c} to \ref{fig:simulation_mean_VaR_ef_n}. The estimates of the mean-VaR efficient frontier are computed for a single simulation run as described at the beginning of this section by using \eqref{eq:bayesian_mean_Q_ef} and \eqref{eq:conventional_mean_Q_ef} for the Bayesian and conventional estimates, respectively. It should be noted that the figures present the most common results which are also observed for other simulation runs. All of the figures also show where the portfolio which globally minimizes the variance is located in the mean-VaR space using each of the methods.

\begin{figure}[H] 
\begin{mdframed}[backgroundcolor=black!10,rightline=false,leftline=false]
  \begin{subfigure}[b]{0.5\linewidth}
    \centering
    \includegraphics[width=0.95\linewidth]{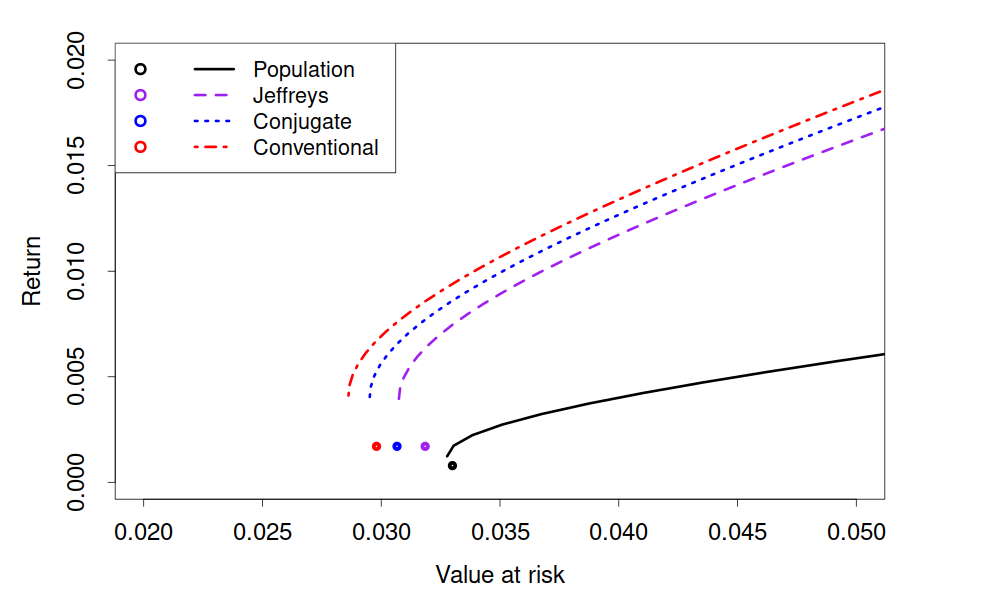} 
    \caption{$c = 0.1$} 
  \end{subfigure}
  \begin{subfigure}[b]{0.5\linewidth}
    \centering
    \includegraphics[width=0.95\linewidth]{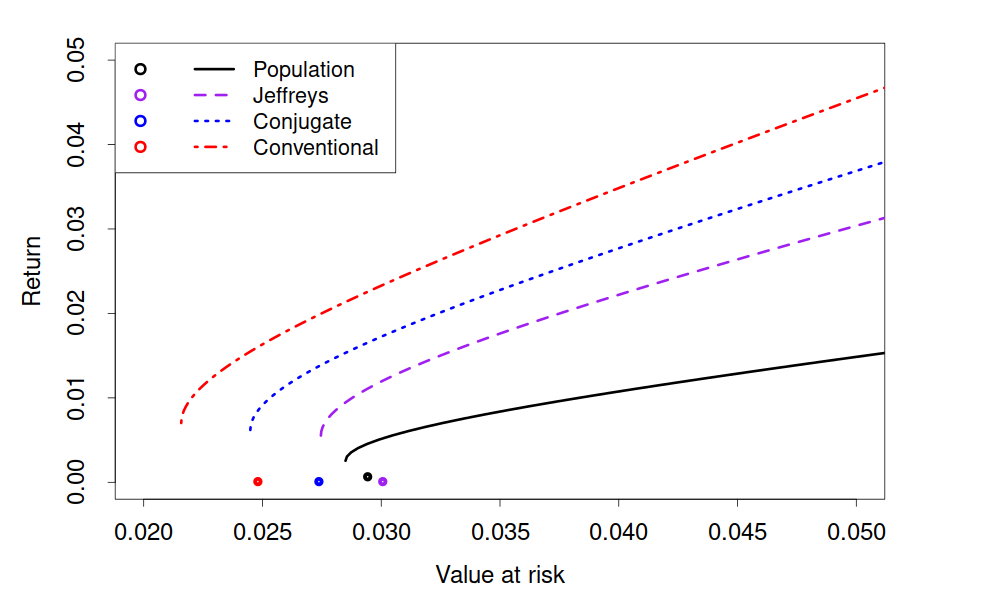}  
    \caption{$c = 0.3$} 
  \end{subfigure} 
  \begin{subfigure}[b]{0.5\linewidth}
    \centering
    \includegraphics[width=0.95\linewidth]{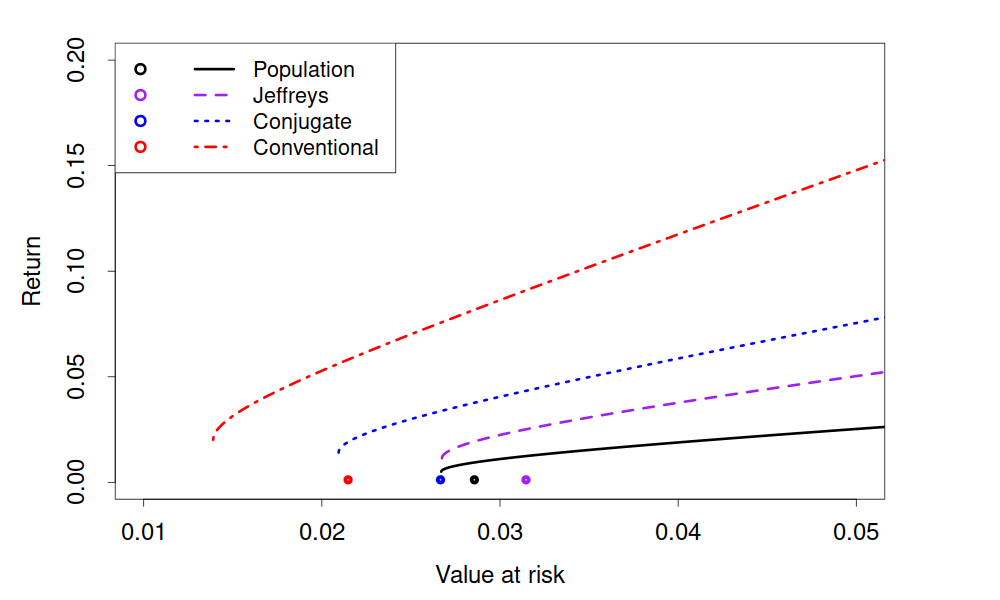}
    \caption{$c = 0.5$} 
  \end{subfigure}
  \begin{subfigure}[b]{0.5\linewidth}
    \centering
    \includegraphics[width=0.95\linewidth]{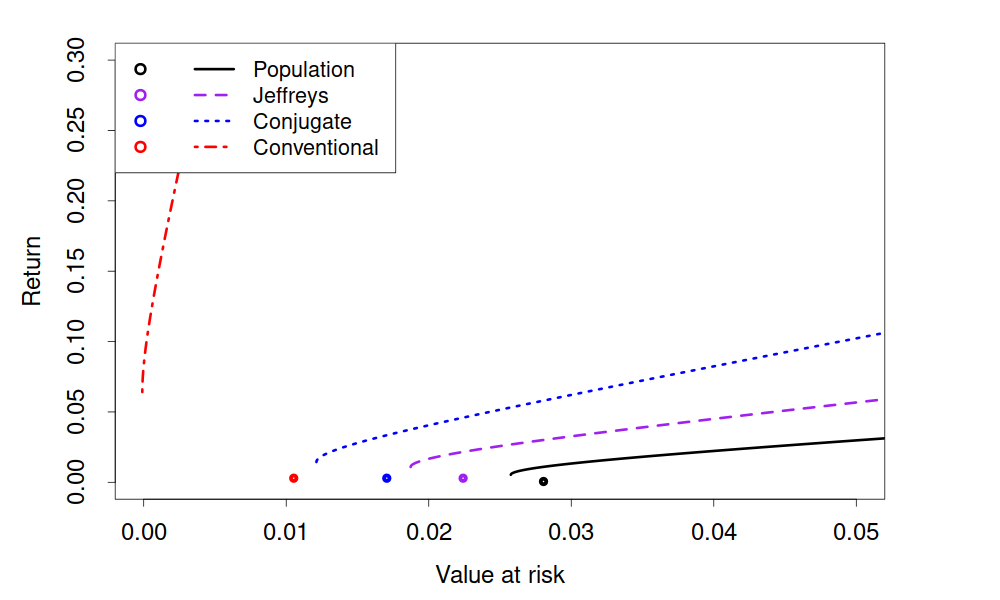}
    \caption{$c = 0.7$} 
  \end{subfigure}
  \caption{Population mean-VaR efficient frontier together with its three estimates for $n=100$, $\alpha = 0.95$ and $c\in\{0.1, 0.3, 0.5, 0.7\}$. The locations of the GMV portfolios are marked by circles. Different scales are used on the axes for presentation purposes.}
 \label{fig:simulation_mean_VaR_ef_c}
\end{mdframed}
\end{figure}

The mean-VaR efficient frontiers and the locations of the GMV portfolios are depicted in Figure \ref{fig:simulation_mean_VaR_ef_c} for $\alpha=0.95$, $n=100$, and $c \in \{0.1,0.3,0.5,0.7\}$. We observe that all methods overestimate the location of the true efficient frontier in the mean-VaR space. Such a behaviour is similar to the one previously documented for the Markowitz efficient frontier in the mean-variance space by \cite{broadie1993computing}, \cite{siegel2007performance}, \cite{bodnar2010unbiased}, \cite{bauder2019bayesian} among others. Namely, ignoring the parameter uncertainty leads to overoptimistic investment opportunities where the investors expect more return for the same level of risk than the population efficient frontier determines. The situation becomes even worse when the conventional mean-VaR frontier is constructed for $c=0.5$ and especially for $c=0.7$. The conventional efficient frontier deviates drastically from the population efficient frontier. Among the two Bayesian efficient frontiers, the one based on the Jeffreys prior leads to the curves that are closest to the population frontier for all considered portfolio sizes. We also observe the positive effect of portfolio diversification in Figure \ref{fig:simulation_mean_VaR_ef_c}. Increasing the portfolio dimension leads to the reduction of the VaR of the GMVaR portfolio. Also, we note the positive effect on the slope parameter of the efficient frontier which becomes larger.

Figure \ref{fig:simulation_mean_VaR_ef_c} also illustrates that the portfolios that minimize the variance are not located on the mean-VaR efficient frontiers. This is an expected but important observation which illustrates that an investor who is mean-variance efficient may not always be mean-VaR efficient.

\begin{figure}[H]
\begin{mdframed}[backgroundcolor=black!10,rightline=false,leftline=false]
  \begin{subfigure}[b]{0.5\linewidth}
    \centering
    \includegraphics[width=0.95\linewidth]{Figures/simulation_mean_VaR_ef_n100_k50_alpha095.png} 
    \caption{$\alpha = 0.95$} 
  \end{subfigure}
  \begin{subfigure}[b]{0.5\linewidth}
    \centering
    \includegraphics[width=0.95\linewidth]{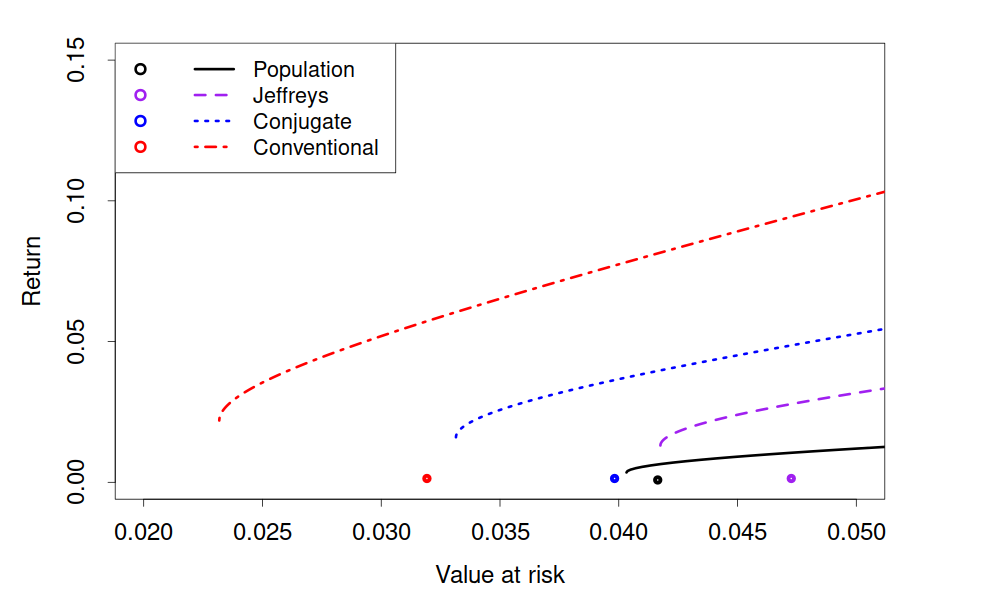}  
    \caption{$\alpha = 0.99$} 
  \end{subfigure}
  \caption{Population mean-VaR efficient frontier together with its three estimates for $n=100$, $c=0.5$ and $\alpha\in\{0.95, 0.99\}$. The locations of the GMV portfolios are marked by circles. Different scales are used on the axes for presentation purposes.}
  \label{fig:simulation_mean_VaR_ef_alpha}
 \end{mdframed}
\end{figure}

\begin{figure}[H]
\begin{mdframed}[backgroundcolor=black!10,rightline=false,leftline=false]
  \begin{subfigure}[b]{0.5\linewidth}
    \centering
    \includegraphics[width=0.95\linewidth]{Figures/simulation_mean_VaR_ef_n100_k50_alpha095.png} 
    \caption{$n = 100$} 
  \end{subfigure}
  \begin{subfigure}[b]{0.5\linewidth}
    \centering
    \includegraphics[width=0.95\linewidth]{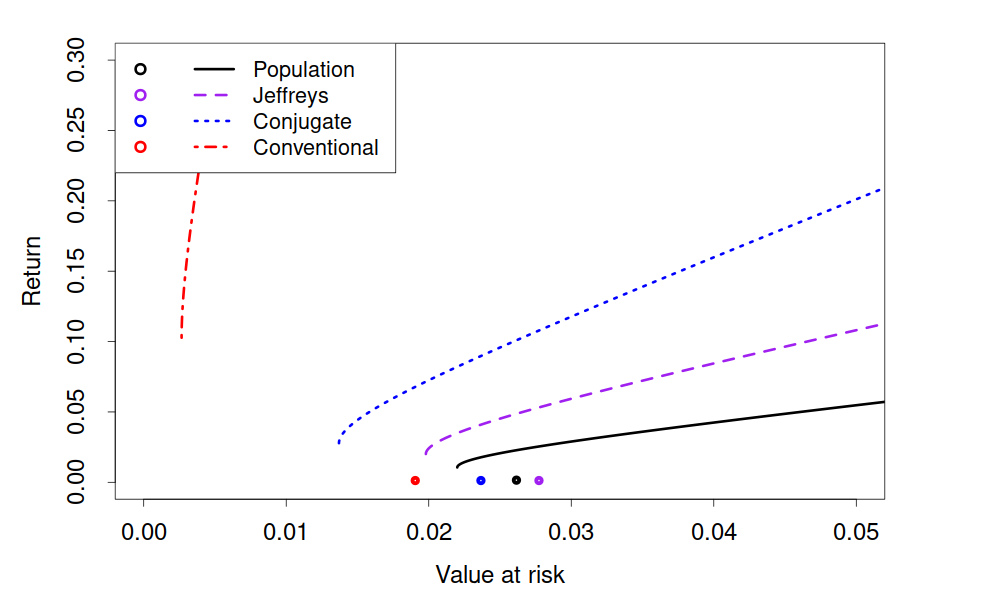}  
    \caption{$n = 200$} 
  \end{subfigure}
  \caption{Population mean-VaR efficient frontier together with its three estimates for $c=0.5$, $\alpha = 0.95$ and $n\in\{100, 200\}$. The locations of the GMV portfolios are marked by circles. Different scales are used on the axes for presentation purposes.}
  \label{fig:simulation_mean_VaR_ef_n}
 \end{mdframed}
\end{figure}

 Figure \ref{fig:simulation_mean_VaR_ef_alpha} and \ref{fig:simulation_mean_VaR_ef_n} demonstrate that the conclusions drawn from the results of Figure \ref{fig:simulation_mean_VaR_ef_c} are also valid for other values of $\alpha$ and $n$. In both figures the Bayesian approach with Jeffreys prior provides the best fit of the population efficient frontier followed by the Bayesian estimate based on the conjugate prior. Also, we observe that the increase of the portfolio dimension with the simultaneous increase of the sample size leads the reduction of the VaR of the GMVaR portfolio and to the increase in the slope parameter of the efficient frontier. Moreover, the GMV portfolios are again shown to not be mean-VaR efficient.
 

\section{Empirical illustration}\label{sec:empirical}
We now continue the comparison between the Bayesian and conventional methodologies through an application on actual market data. As in the simulation study, we study the performance and existence of the GMVaR portfolio and investigate the behaviour of their efficient frontiers. Once again, we consider the cases $n \in \{100, 200\}$, $k = cn$ for $c \in \{0.1, 0.3, 0.5, 0.7\}$ and $\alpha \in \{0.95, 0.99\}$. 

\subsection{Data description}\label{sec:data_emp}
We use weekly returns on stocks included in the S\&P 500 index for the period from the 1st of January, 2010 to the 28th of March, 2020.
In order to circumvent the possible bias of selecting stocks which outperform or underperform the rest of the market, we consider all stocks included in the S\&P 500 index by our end date that were already part of the the index by our chosen start date. The lack of public information makes it difficult to know exactly when a certain stock was added to this index, but based on \cite{wiki:SP500} we have chosen to consider 221 stocks that were present in the index before the 1st of January, 2010. A complete list of the stocks is provided in Table~\ref{table:stocks} in  Appendix~\ref{appendix:stocks}. 

In the empirical analysis, we randomly choose 500 portfolios of size $k$ from the list of stocks for each possible value of $\{n,k,\alpha\}$. Once the stocks have been selected, they are kept for the whole time period but the weights of the GMVaR portfolio are re-calculated each week and the performance is evaluated on a weekly basis and then averaged across all sampled portfolios. It should be noted that the performance results are only based on portfolios which satisfy the conventional and Bayesian GMVaR conditions \eqref{eq:conventional_GMQ_condition} and \eqref{eq:bayesian_GMQ_condition}, respectively, for all estimates of the GMVaR portfolio simultaneously. Finally, the hyperparameters  when using the conjugate prior are specified as in the simulation study, i.e., by employing the empirical Bayesian approach and setting $d_0=r_0=n$.

\subsection{Results of the empirical illustration}\label{sec:empirical_GMVaRP}

Regarding the existence of the estimates of the GMVaR portfolio, it is more likely that the Bayesian GMVaR condition \eqref{eq:bayesian_GMQ_condition} is satisfied than the conventional condition \eqref{eq:conventional_GMQ_condition}. The Bayesian GMVaR portfolios exist all the time whereas the conventional GMVaR portfolio does not always exist when $\{n = 100, k = 70, \alpha = 0.95\}$, $\{ n = 200, k = 140, \alpha = 0.95\}$ and $\{n = 100, k = 70, \alpha = 0.99\}$. For those values, the conventional GMVaR condition fails during the time period for 369, 5 and 1 portfolios, respectively, out of the 500 portfolios. Hence, we observe that the conditions are more likely to be satisfied when $\alpha$ is large or when $c$ is small.

As in the simulation study, we consider the relative VaR exceedance frequency when evaluating the performance of each estimate of the GMVaR portfolio. This value should ideally be close to $1-\alpha$. The result is summarized in Table \ref{tab:empirical_GMVaRP_VaR_exceedance}.

\begin{singlespace}
\begin{table}[H]
\small
\begin{mdframed}[backgroundcolor=black!10,rightline=false,leftline=false]
\centering
\caption{Relative VaR exceedance frequencies for the three estimates of the GMVaR portfolio.}
\label{tab:empirical_GMVaRP_VaR_exceedance}
{\begin{tabular}{| c | c | c | c | c |c |}
  \hline
  \multicolumn{3}{|c|}{Parameter setup}
  & \multicolumn{3}{c|}{GMVaR portfolio} \\
  \hline
  $\alpha$ & $n$ & $k$
   & Jeffreys & Conjugate & Conventional \\
  \hhline{|=|=|=|=|=|=|}
    \multirow{8}{*}{$0.95$} & 
    \multirow{4}{*}{$100$} &
    $10$ & 0.0771 & 0.0828 & 0.0872 \\ \cline{3-6}
    & & 
    $30$ & 0.1140 & 0.1356 & 0.1602 \\ \cline{3-6}
    & & 
    $50$ & 0.1617 & 0.2071 & 0.2698 \\ \cline{3-6}
    & &
    $70$ & 0.2227 & 0.2900 & 0.4040 \\
  \cline{2-6}
    &\multirow{4}{*}{$200$} & 
    $20$ & 0.0794 & 0.0846 & 0.0890 \\ \cline{3-6}
    & &
    $60$ & 0.1144 & 0.1354 & 0.1600 \\ \cline{3-6}
    & &
    $100$ & 0.1625 & 0.2047 & 0.2655 \\ \cline{3-6}
    & &
    $140$ & 0.2294 & 0.2920 & 0.3966 \\
  \hline
   \multirow{8}{*}{$0.99$} & 
    \multirow{4}{*}{$100$} &
    $10$ & 0.0361 & 0.0394 & 0.0424 \\ \cline{3-6}
    & & 
    $30$ & 0.0532 & 0.0664 & 0.0833 \\ \cline{3-6}
    & & 
    $50$ & 0.0802 & 0.1137 & 0.1669 \\ \cline{3-6}
    & &
    $70$ & 0.1283 & 0.1911 & 0.3057 \\
  \cline{2-6}
    &\multirow{4}{*}{$200$} & 
    $20$ & 0.0375 & 0.0404 & 0.0429 \\ \cline{3-6}
    & &
    $60$ & 0.0521 & 0.0646 & 0.0818 \\ \cline{3-6}
    & &
    $100$ & 0.0816 & 0.1135 & 0.1650 \\ \cline{3-6}
    & &
    $140$ & 0.1359 & 0.1948 & 0.3046 \\
    \hline
\end{tabular}}{}
\end{mdframed}
\end{table}
\end{singlespace}
As can be seen, the Bayesian approaches are clearly outperforming the conventional method by having exceedance frequencies closer to $1-\alpha$. Using the Jeffreys prior gives the best results in all situations that we consider. However, as in the simulation study, all of the methods are underestimating VaR since the exceedance frequency is always higher than $1-\alpha$. This is especially pronounced when $c$ is large, i.e., in the case of a large-dimensional portfolio. Even if none of the methods performs very well for such situations, the Bayesian approaches, especially the one based on the Jeffreys prior, provides a considerable improvement in comparison to the conventional method by reducing the relative exceedance frequency by 45\% for $\alpha=0.95$ and by 55\% for $\alpha=0.99$ when $k/n=0.7$. The application of the Bayesian approach based on the conjugate prior also results in much lower relative exceedance frequencies compared to the conventional method, although this Bayesian GMVaR portfolio performs always worse than the one based on the Jeffreys prior.

Similar results to those observed in Table~\ref{tab:empirical_GMVaRP_VaR_exceedance} are also present in Figure \ref{fig:empirical_mean_VaR_ef_c} where we plot the estimated mean-VaR efficient frontier for the end date using $\alpha=0.95$, $n=100$, and $c=k/n \in \{0.1,0.3,0.5,0.7\}$. 
Both Bayesian efficient frontiers are always located under the conventional efficient frontier. While the three efficient frontiers almost coincide when $c=0.1$, the difference between the conventional and Bayesian approaches becomes pronounced when $c$ becomes larger, particularly when $c=0.7$. Moreover, we again see that the GMV portfolios are not mean-VaR efficient. All of this is in line with the observations made for Figure \ref{fig:simulation_mean_VaR_ef_c} in the simulation study. Varying $\alpha$ and $n$ using the empirical data will also result in the same relationships between the efficient frontiers as shown in Figures \ref{fig:simulation_mean_VaR_ef_alpha} and \ref{fig:simulation_mean_VaR_ef_n} in the simulation study, indicating the considerable overoptimism present in the construction of the conventional efficient frontier.

\begin{figure}[H]
\begin{mdframed}[backgroundcolor=black!10,rightline=false,leftline=false]
  \begin{subfigure}[b]{0.5\linewidth}
    \centering
    \includegraphics[width=0.95\linewidth]{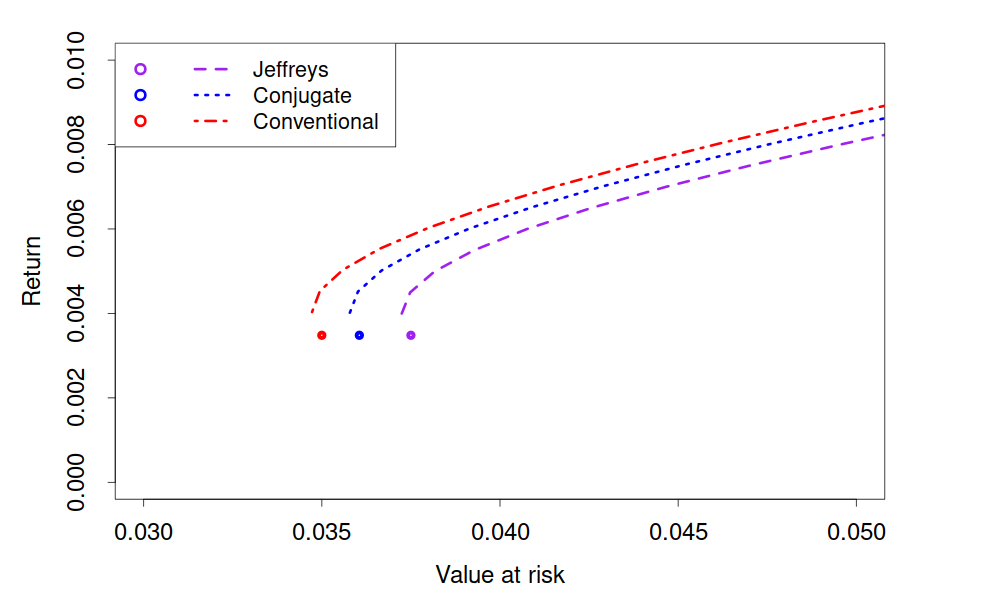}
    \caption{$c = 0.1$} 
  \end{subfigure}
  \begin{subfigure}[b]{0.5\linewidth}
    \centering
    \includegraphics[width=0.95\linewidth]{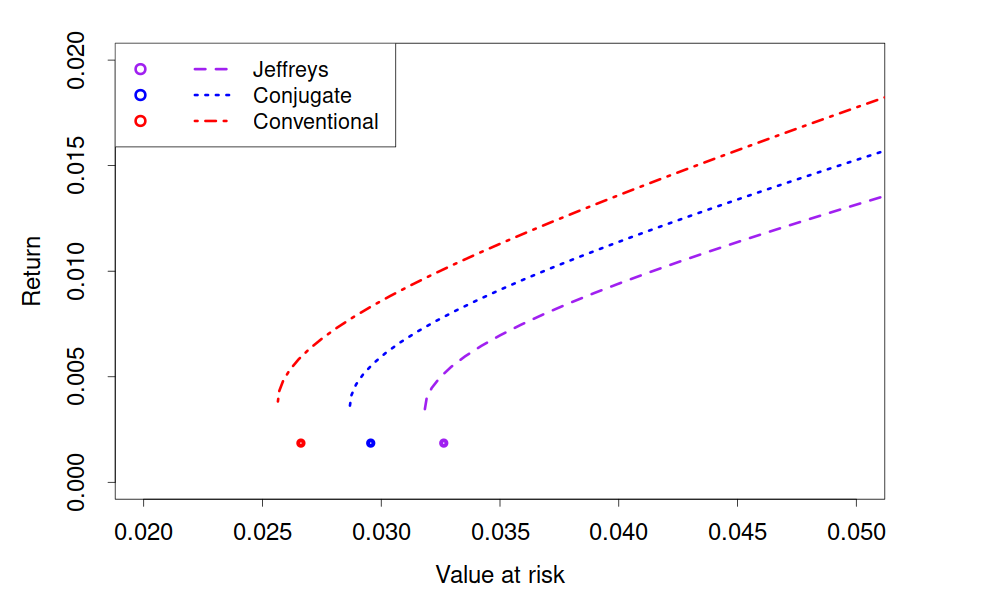}
    \caption{$c = 0.3$}
  \end{subfigure} 
  \begin{subfigure}[b]{0.5\linewidth}
    \centering
    \includegraphics[width=0.95\linewidth]{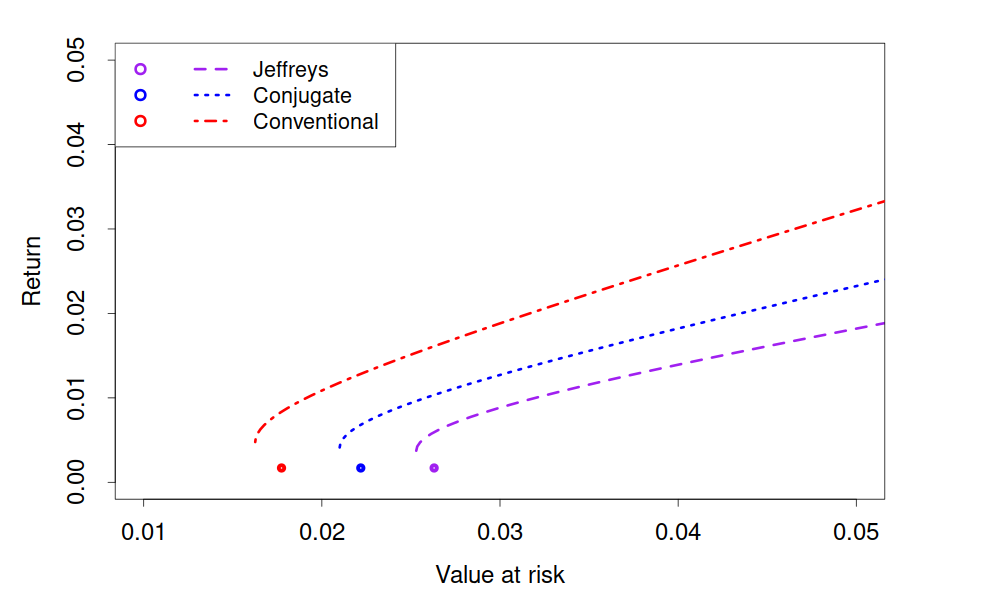}
    \caption{$c = 0.5$}
  \end{subfigure}
  \begin{subfigure}[b]{0.5\linewidth}
    \centering
    \includegraphics[width=0.95\linewidth]{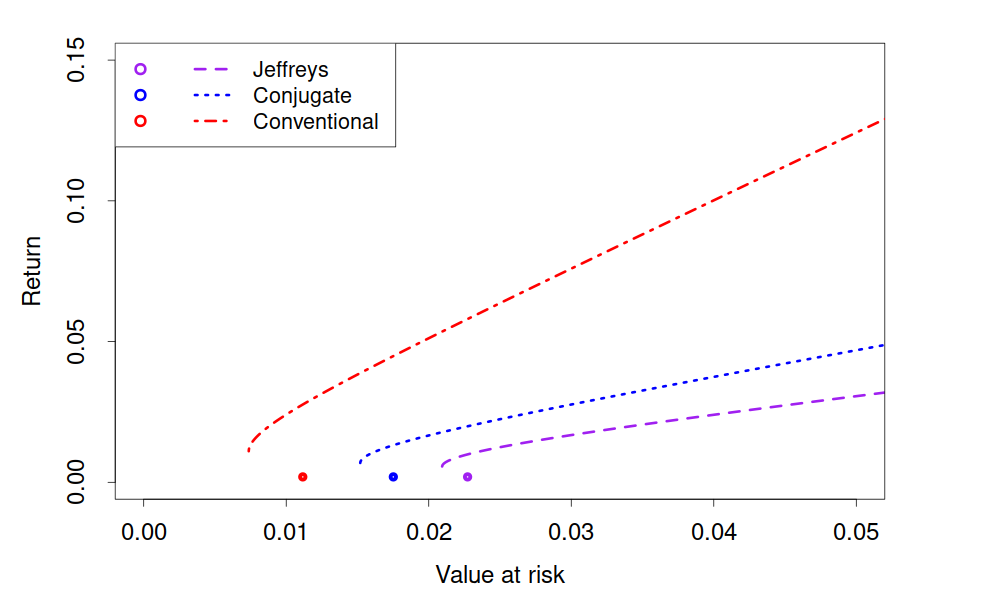}
    \caption{$c = 0.7$}
  \end{subfigure}
    \caption{Bayesian and conventional mean-VaR efficient frontiers based on empirical data for $n=100$, $\alpha = 0.95$ and $c\in\{0.1, 0.3, 0.5, 0.7\}$. The locations of the GMV portfolios are marked by circles. Different scales are used on the axes for presentation purposes.}
  \label{fig:empirical_mean_VaR_ef_c} 
\end{mdframed}
\end{figure}

\section{Conclusion}\label{sec:summary}
The traditional mean-variance analysis has been a paramount foundation for the extension to portfolio analysis based on one-sided risk measures which are popular in financial mathematics. However, the conventional approaches related to the construction of optimal portfolios usually ignore the parameter uncertainty in the construction of an optimal portfolio. It is common to define optimal portfolios by a two-step procedure where first an optimization problem is solved and then the optimal portfolios are estimated by replacing the unknown quantities in the solutions by the corresponding sample counterparts. 

The Bayesian methodology pose a fundamental difference to the conventional approaches in its viewpoint on what we want to optimize: Investors care about their \textit{future risk in taking a position}, not the risk of having the a certain position today. In light of data, today's outcome is already determined and usually not interesting. The Bayesian framework use the predictive posterior distribution to cope with this. That is, the Bayesian methodology answers the problem in a straightforward manner while the conventional method simply ignores it.

We contribute to the existent literature by formulating and solving quantile-based portfolio allocation problems from the perspective of Bayesian statistics. This approach is advantageous since it allows to take the parameter uncertainty into account before the optimization problem is solved. The development of the general risk functionals from the Bayesian perspectives appears to be a very promising subject of research with great potential of future development. The risk functionals can be defined through all information available up to the point in time when a portfolio is constructed or a decision on the risk of the current position should be made. 
As a result, no unknown or unobservable quantities are present in their definitions. This is a very appealing property since it takes all uncertainties into account before the risk functional is determined.

In the frequentist setting, the general risk measure of an optimal portfolio choice problem will have a similar structure as under the Bayesian setup when asset returns are elliptically contoured distributed \citep[see, e.g., ][for a definition and properties thereof]{gupta2013elliptically}. However, both the portfolio expected return and the portfolio variance are determined by unknown parameters of the distribution of asset returns which must be estimated in any practical application. This procedure would lead to an important task, namely to include the parameter uncertainty in the definition of the general risk measure. This challenging task has not properly been treated in the literature up to now when frequentist methods are employed, while the Bayesian approach provides an intelligent automatic solution. 

Results of the simulation study and of the empirical application leads to the conclusion that the Bayesian approaches to portfolio construction provide a good alternative to the conventional procedures and they are usually preferable in most of the considered cases. The Bayesian approaches outperform the conventional one in terms of providing a better VaR prediction. This holds uniformly, independently of the portfolio dimension, sample size, and the confidence level used in the computation of the VaR. Only when the portfolio dimension is relatively small to the sample size does the conventional method perform similarly to the Bayesian approaches. Such a behavior is expected since the priors used in the derivation of Bayesian inference can be interpreted as a regularisation and it might not be necessary employ that in such cases. Although using Jeffreys prior gave the best results in our study, a more careful calibration of the hyperparameters of the conjugate prior could have made that one more beneficial. In practice, the hyperparamters would be specified using knowledge from experts within fundamental market analysis. 

We also find that the conventional mean-VaR efficient frontier considerably overestimates the location of the true mean-VaR frontier. Although the Bayesian approaches reduce the underestimation of the VaR considerably and shrink the estimates of the efficient frontier, they still show significant overoptimism when the portfolio dimension is large in comparison to the sample size, i.e., when a large-dimensional optimal portfolio is constructed. Further research in this direction is needed which might lead to interesting results completing the existing findings in the direction of large-dimensional portfolio construction \citep[see, e.g.,][]{fan2012vast, hautsch2015high, BodnarDmytriv2019, cai2020high}.

\bibliography{bayesian_quantile}

\section*{Acknowledgement} This research was partly supported by the Swedish Research Council (VR) via the project ``Bayesian Analysis of Optimal Portfolios and Their Risk Measures''.

\appendix
\section{Proofs of theoretical results}\label{appendix:proofs}

In order to get the stochastic representations presented below we also need the following result.
\begin{lemma}\label{prop:sum_of_t_dist}
Let a random variable $z$ possess the following stochastic representation
\begin{equation}
    z \eqdist \frac{\tau_1}{\sqrt{vd}}+\sqrt{1+\frac{\tau_1^2}{d}}\frac{\tau_2}{\sqrt{d+1}},
    \label{eq:z_expression}
\end{equation}
where $d>0$, $\tau_1$ and $\tau_2$ are independent with $\tau_1 \sim t(d)$ and $\tau_2 \sim t(d+1)$. Then, $z$ follows a $t$-distribution with $d$ degrees of freedom, location parameter 0, and scale parameter $\sqrt{(v+1)/vd}$.
\end{lemma}

\begin{proof}[Proof of Lemma \ref{prop:sum_of_t_dist}.] 
Since $\tau_1$ and $\tau_2$ are independent with $\tau_2 \sim t(d+1)$, the conditional distribution of $z$ given $\tau_1$ is a $t$-distribution with $d+1$ degrees of freedom, location parameter ${\tau_1}/{\sqrt{vd}}$ and scale parameter ${g(\tau_1)}/{\sqrt{d+1}}$ with $g(\tau_1)=\sqrt{1+{\tau_1^2}/{d}}$. Thus the joint distribution of $z$ and $\tau_1$ is given by
\begin{align*}
    f(z,\tau_1) & = f(z|\tau_1)f(\tau_1) \\
    & = \frac{\Gamma\left(\frac{d+2}{2}\right)}{\Gamma\left(\frac{d+1}{2}\right)}
    \frac{1}{\sqrt{\pi}}\frac{1}{g(\tau_1)}\left(1+
    \left(\frac{z-\frac{\tau_1}{\sqrt{vd}}}{g(\tau_1)}\right)^2\right)^{-\frac{d+2}{2}}
    \frac{\Gamma\left(\frac{d+1}{2}\right)}{\Gamma\left(\frac{d}{2}\right)}
    \frac{1}{\sqrt{\pi d}}\left(1+\frac{\tau_1^2}{d}\right)^{-\frac{d+1}{2}} \\
    & \propto \frac{1}{g(\tau_1)}
    \left(1+\left(\frac{z-\frac{\tau_1}{\sqrt{vd}}}{g(\tau_1)}\right)^2\right)^{-\frac{d+2}{2}}
    g(\tau_1)^{-(d+1)}
     = \left(g(\tau_1)^2+\left( z-\frac{\tau_1}{\sqrt{vd}}\right)^2\right)^{-\frac{d+2}{2}} \\
    & = \left(1+\frac{1}{d}[z, \tau_1] \begin{bmatrix} d & -\frac{\sqrt{vd}}{v} \\ -\frac{\sqrt{vd}}{v} & \frac{v+1}{v}\end{bmatrix} \begin{bmatrix} z \\ \tau_1\end{bmatrix}\right)^{-\frac{d+2}{2}}.
\end{align*}

The last expression is the kernel of a multivariate $t$-distribution with $d$ degrees of freedom, location vector $\bol{\nu}=\boldsymbol{0}$ and dispersion matrix $\bOmega$ given by
\begin{equation*}
\bOmega = \begin{bmatrix} \frac{v+1}{vd} & \frac{\sqrt{vd}}{vd} \\ \frac{\sqrt{vd}}{vd} & 1 \end{bmatrix}.    
\end{equation*}
Hence, the marginal distribution of $z$ is also a $t$-distribution with $d$ degrees of freedom, location $0$ and scale $\sqrt{(v+1)/vd}$ \citep[see, e.g.,][]{kotz2004multivariate}.
\end{proof}

\begin{proof}[Proof of Theorem~\ref{theorem:predictive_posterior}.]
\cite{bauder2018bayesian} characterized the posterior predictive distribution of the portfolio return by deriving the stochastic representation of $\widehat{X}_{P,t}$ given by
\begin{equation*}
 \widehat{X}_{P,t}\eqdist  
 m+\sqrt{s }
 \left(\frac{\tau_1}{\sqrt{vd}}+\sqrt{1+\frac{\tau_1^2}{d}}\frac{\tau_2}{\sqrt{d+1}}\right)
 \end{equation*}
 with $m=\mathbf{w}^T\bar{\bx}_{t-1,J}$, $s=\mathbf{w}^T\bS_{t-1,J}\mathbf{w}$, $v=n$, and $d=n-k$ under the Jeffreys prior and with $m=\mathbf{w}^T\bar{\bx}_{t-1,I}$, $s=\mathbf{w}^T\bS_{t-1,I}\mathbf{w}$, $v=n+r_0$ and $d=n+d_0-2k$ under the conjugate prior. The application of Lemma~\ref{prop:sum_of_t_dist} leads to the statement of the theorem.
\end{proof}

\begin{proof}[Proof of Theorem \ref{theorem:convexity_of_VaR_and_CVaR}.]
The statement of the theorem follows from the fact that $\bw^\top \bar{\bx}_{t-1}$ is linear in $\bw$ and that, since $\bS_{t-1}$ is positive definite, $\sqrt{\bw^T\bS_{t-1}\bw}$ can be regarded as the Euclidean norm of $\bS_{t-1}^{1/2}\bw$ where $\bS_{t-1}^{1/2}$ is the symmetric square root of $\bS_{t-1}$. Since $q_{\alpha}>0$ and  the Euclidean norm is convex, the result
follows.
\end{proof}

\begin{proof}[Proof of Theorem \ref{theorem:bayesian_GMQ}.]
Let $c_{k,n}=\frac{d_{k,n}r_{k,n}}{d_{k,n}-2}$. Since the solution of 
\begin{equation*}
\min_{\bw:\,\bw^\top\mathbf{1}=1} Q_{t-1}(\bw),
\end{equation*}
belongs to the Bayesian efficient frontier \eqref{eq:bayesian_ef_mv} in the mean-variance space, it can be found by solving the univariate optimization problem given by
\begin{equation}\label{eq:optimal_V}
\min_{V: V\ge V_{GMV,t-1}} -R_{GMV,t-1}-(c_{k,n})^{-1/2}\sqrt{s_{t-1}}\sqrt{V-V_{GMV,t-1}} + q_{\alpha} (c_{k,n})^{-1/2}\sqrt{r_{k,n}} \sqrt{V}
\end{equation}
where $R_{GMV,t-1}$ and $V_{GMV,t-1}$ are given in \eqref{eq:bayesian_gmv_weights} and $s_{t-1}$ is defined in \eqref{eq:bayesian_ef_mv}. The solution of \eqref{eq:optimal_V} solves
\begin{equation}\label{eq:derivative_of_global_min_VaR_CVaR_objective}
    q_{\alpha} \sqrt{r_{k,n}}\frac{1}{\sqrt{V}}=\sqrt{s_{t-1}}\frac{1}{\sqrt{V-V_{GMV,t-1}}}
\end{equation}
and it is given by
\begin{equation}\label{eq:var_of_global_min_VaR_CVaR}
    V_{GMQ,t-1} = \frac{q_{\alpha}^2}{q_{\alpha}^2-r_{k,n}^{-1}s_{t-1}}V_{GMV,t-1}.
\end{equation}
where it obviously holds that $V_{GMQ,t-1}>V_{GMV,t-1}$ as soon as $q_{\alpha}^2-r_{k,n}^{-1}s_{t-1}>0$, which coincides with the second order condition needed to ensure that $V_{GMQ,t-1}$ is the solution of \eqref{eq:optimal_V}.

Finally, $R_{GMQ,t-1}$ is obtained from \eqref{eq:bayesian_ef_mv} and it is given by 
\begin{eqnarray*}
    R_{GMQ,t-1} &=& R_{GMV,t-1}+\sqrt{c_{k,n}^{-1}s_{t-1}} \sqrt{\frac{q_{\alpha}^2}{q_{\alpha}^2-r_{k,n}^{-1}s_{t-1}}V_{GMV,t-1}-V_{GMV,t-1}} \\
    &=&R_{GMV,t-1}+ \frac{r_{k,n}^{-1}s_{t-1}}{\sqrt{q_{\alpha}^2-r_{k,n}^{-1}s_{t-1}}}\sqrt{\frac{d_{k,n}-2}{d_{k,n}}}\sqrt{V_{GMV,t-1}}.
\end{eqnarray*}
\end{proof}

\begin{proof}[Proof of Theorem \ref{theorem:bayesian_mean_Q_ef}.]
From \eqref{eq:bayesian_ef_mv} and \eqref{eq:VaR_CVaR_common}, we get
\begin{equation}\label{eq_app1}
\frac{(R-R_{GMV,t-1})^2}{a_{t-1}}+V_{GMV,t-1}=V
\quad\text{and}\quad    
V=\left(\frac{R+Q}{b}\right)^2
\end{equation}
where
\begin{equation*}
a_{t-1}= \frac{d_{k,n}-2}{d_{k,n}r_{k,n}}s_{t-1}   
\quad\text{and}\quad    
b=q_{\alpha}\sqrt{\frac{d_{k,n}-2}{d_{k,n}}}.
\end{equation*}
From \eqref{eq_app1}, we get
\begin{eqnarray*}
Q&=&b\sqrt{\frac{(R-R_{GMV,t-1})^2}{a_{t-1}}+V_{GMV,t-1}}-R\\
&=&q_{\alpha}\sqrt{\frac{(R-R_{GMV,t-1})^2}{r_{k,n}^{-1}s_{t-1}}
+\frac{d_{k,n}-2}{d_{k,n}}V_{GMV,t-1}}-R.
\end{eqnarray*}

Finally, we note that the Bayesian efficient frontier in the mean-Q space can be rewritten as
\begin{equation*}
R^2-2RR_{GMV,t-1}+R_{GMV,t-1}^2 - \frac{a_{t-1}}{b^2}R^2
-2\frac{a_{t-1}}{b^2}RQ-\frac{a_{t-1}}{b^2}Q^2
+a_{t-1}V_{GMV,t-1}    =0,
\end{equation*}
which is a hyperbola in the mean-Q space for $s_{t-1}>0$ \citep[see, e.g., Section 3.5.2.11 in][]{bronshtein2013handbook} since
\[-\frac{a_{t-1}}{b^2}\left(1-\frac{a_{t-1}}{b^2}\right)-\frac{a_{t-1}^2}{b^4}=-\frac{a_{t-1}}{b^2}=
-\frac{s_{t-1}}{r_{k,n}q_{\alpha}^2} <0.\]
\end{proof}

\section{List of stocks}\label{appendix:stocks}
Table \ref{table:stocks} presents the list of stocks considered in Section \ref{sec:empirical}.
\begin{singlespace}
\begin{table}[H]
\small
\begin{mdframed}[backgroundcolor=black!10,rightline=false,leftline=false]
\centering
\caption{Stocks considered in the empirical illustration.}
\label{table:stocks}
{\begin{tabular}{|c c c c c c c c c|}
    \hline
 MMM & ABT & ADBE & AES & AFL & A  &  APD & AKAM & ALL \\
 \hline
 GOOG & AMZN & AEE & AXP & AIG & AMT & AMP & ABC &  AMGN \\
 \hline
 APH & ADI & ANTM & AON &  APA & AIV & AAPL  & AMAT & ADM \\
 \hline
 AIZ & T &  ADSK & ADP &  AZO & AVB & AVY & BLL & BAC \\
 \hline
 BK  & BAX &  BDX & BBY & BIIB & BKNG & BXP & BSX & BF.B \\
 \hline
 CHRW & COG & COF & CAH & CCL & CBRE & CNP & CTL & CF \\
 \hline
 SCHW & CI & CINF & CTAS & CSCO & C & CTXS & CLX & CME \\
 \hline
 CMS & CTSH & CMA & CAG & STZ & COST & CSX & CMI & DVA \\
 \hline
 XRAY & DVN & DFS & DOV & DUK & ETFC & EMN & ECL & EA \\
 \hline
 EMR  & EOG & EFX & EQR & EL  & EXPE & EXPD & FAST & FDX \\
 \hline
 FIS & FISV & FLIR & FLS &  FMC & GPS & GIS & GPC & GILD \\
 \hline
 GL & GS &  GWW & HRB & HAS & PEAK & HES & HD &  HON \\
 \hline
 HRL & HST & HPQ & ITW & INTC & ICE & IPG & IFF &  INTU \\
 \hline
 ISRG &  IVZ & IRM & J & SJM & JNJ & JPM & JNPR & KEY \\
 \hline
 KIM &  LB & LH &  LEN & LLY &  LNC & LIN & LMT & LOW \\
 \hline
 MRO & MMC & MAS & MA &  MCD & MDT & MCHP & MU  & MSFT \\
 \hline
 TAP & MYL & NDAQ & NOV &  NTAP & NWL & NEM & NEE & NKE \\
 \hline
 NBL &  JWN  & NOC & NUE & NVDA &  ORLY & OXY & ORCL & PCAR \\
 \hline
 PH  & PBCT & PKI & PM &  PXD & PNC & PFG & PGR & PLD \\
 \hline
 PRU & PSA & PHM &  PWR & DGX & RL &  RF &  RSG & RHI \\
 \hline
 ROP & ROST & CRM & SLB & SHW & SPG & SNA & LUV &  SWK \\
 \hline
 SYK & SYY & TGT & FTI & TXT & TIF &  TJX & TRV & TFC \\
 \hline
 UNH & UPS & UNM & VFC & VAR & VTR & VRSN & VZ &  V \\
 \hline
 VMC & WMT & WBA & DIS & WEC & WFC  & WELL & WDC & WMB \\
 \hline
 WYNN &  XLNX & YUM &  ZBH & ZION & & & & \\
 \hline
    \end{tabular}}{}
\end{mdframed}
\end{table}
\end{singlespace}
\end{document}